\documentclass[journal]{IEEEtran}
\usepackage[cmex10]{amsmath}
\usepackage{amssymb}
\usepackage{color}

\usepackage{mathrsfs}
\usepackage{amsfonts}
\usepackage{float}

\usepackage{amsfonts}

\newcommand{\wt} {$\text{wt}$}

\newcommand{\spa}{\textrm{span}}
\ifCLASSINFOpdf
\else
\fi
\begin{document}
%
\title{On $2k$-Variable Symmetric Boolean Functions
with Maximum Algebraic Immunity $k$}
%
%
%

\author{Hui~Wang,
        Jie~Peng,
        Yuan~Li,
        and~Haibin Kan
\thanks{H. Wang is with Shanghai Key Lab of Intelligent Information Processing, School of Computer Science, Fudan University, Shanghai 200433, P. R. China
 (e-mail: 1011024022@fudan.edu.cn).}
\thanks{J. Peng is with School of Mathematics and Statistics, Huazhong Normal University, Wuhan 430079, P. R. China
 (e-mail: jiepeng@mail.ccnu.edu.cn).}
\thanks{Y. Li and H. Kan are with Shanghai Key Lab of Intelligent Information Processing, School of Computer Science, Fudan University, Shanghai 200433, P. R. China
 (e-mails: 07300720173@fudan.edu.cn; hbkan@fudan.edu.cn).}
}

\maketitle

\begin{abstract}
Algebraic immunity of Boolean function $f$ is defined as the minimal
degree of a nonzero $g$ such that $fg=0$ or $(f+1)g=0$. Given a
positive even integer $n$, it is found that the weight distribution
of any $n$-variable symmetric Boolean function with maximum
algebraic immunity $\frac{n}{2}$ is determined by the binary
expansion of $n$. Based on the foregoing, all $n$-variable symmetric
Boolean functions with maximum algebraic immunity are constructed.
The amount is $(2\wt(n)+1)2^{\lfloor \log_2 n \rfloor}$.
\end{abstract}

\begin{IEEEkeywords}
Algebraic attack, algebraic immunity, symmetric Boolean function.
\end{IEEEkeywords}

%
\IEEEpeerreviewmaketitle

\section{Introduction}
\IEEEPARstart{A}{lgebraic} attacks have received much attention in
cryptographic analyzing stream and block cipher systems
\cite{:Armknecht}, \cite{1:Courtois}, \cite{:S.Ronjom}, which try to
recover the secret key by solving overdefined systems of
multivariate equations. Therefore, algebraic immunity ($\text{AI}$),
a new cryptographic property for designing Boolean functions, was
proposed by W. Meier {\sl et al}. $\cite{:Meier }$. Algebraic
immunity of the Boolean function used in a cryptosystem should be
high enough to resist algebraic attacks. The upper bound of the
algebraic immunity of an $n$-variable Boolean function is $\lceil
\frac{n}{2} \rceil$ $\cite{1:Courtois, :Meier }$. Several
theoretical constructions of Boolean functions with optimal
$\text{AI}$ have been presented in the literature $\cite{:Carlet}$,
$\cite{1:Dalai}$, \cite{5:keqing Feng}, \cite{1:Qu}.

Symmetric Boolean functions are of great interest from a
cryptographic point of view. An $n$-variable symmetric Boolean
function can be identified by an $(n+1)$-bit vector, so symmetric
Boolean functions have smaller hardware size than average Boolean
functions. They allow the computation of values for functions with
more variables than general ones. For this reason, symmetric Boolean
functions have been paid particular attention.

For an odd integer $n$, Dalai {\sl et al}. showed that a Boolean
function with maximum $\text{AI}$ should be balanced \cite{1:Dalai}.
In \cite{:Lina}, it was proved that the majority function
$\text{Maj}_n$ and its complement $\text{Maj}_n+1$ are the only two
trivially balanced symmetric Boolean functions with maximum
$\text{AI}$. It also has been proven that the number of symmetric
Boolean functions with maximum $\text{AI}$ is exactly two
\cite{A:Qu}.

For the case where $n$ is even, the situation becomes very
complicated. A few classes of even-variable symmetric Boolean
functions with maximum $\text{AI}$ have been constructed in
\cite{1:Qu}, \cite{2:Braeken}. However, only the number and form of
$2^m$-variable symmetric Boolean functions with maximum algebraic
immunity have been solved by introducing the weight support
technique \cite{5:keqing Feng}. This method has also been used to
determine the number of $(2^m+1)$-variable symmetric Boolean
functions with submaximum algebraic immunity $2^{m-1}$ \cite{:Liao}.

In this paper, we first study the weight distribution of those
$n$-variable symmetric Boolean functions achieving maximum algebraic
immunity with $n$ even. We find that the set $N=\{0,$ $1,$ $\ldots,$
$n\}$ can be divided into some particular subsets according to the
binary expansion of $n$, on which the Boolean functions should be
constant. Meanwhile, the values of the functions on these subsets
should satisfy some strict conditions. Furthermore, we continue to
prove that all the symmetric Boolean functions constructed following
the above laws indeed achieve maximum algebraic immunity. Thus, we
construct all the even-variable symmetric Boolean functions with
maximum algebraic immunity. The number of these functions and their
corresponding hamming weights are also obtained.

The organization of the paper is as follows. In the following
section, we present some basic notations and knowledge about Boolean
functions. In section 3, we obtain some necessary conditions for an
even-variable symmetric Boolean function to reach maximum algebraic
immunity. In the next two sections, we prove that these conditions
are sufficient. The main theorem of this paper is given in section
6. Section 7 concludes the paper.





\section{Preliminaries}
\vspace{0.3cm} Let $\text{F}_{2}^{n}$ be the $n$-dimensional vector
space over the finite field $\text{F}_2$, and $e_0^n$, $e_1^n$,
$\ldots,$ $e_{n-1}^n$ be its normal basis,
$$e_0^n=(1, 0, \ldots, 0), e_1^n=(0, 1, \ldots, 0),..., e_{n-1}^n=(0, 0, \ldots, 1).$$
The superscript $n$ may be omitted if there is no confusion.

An $n$-variable Boolean function is a function from $\text{F}_{2}^n$
into $\text{F}_2$. Let $B_{n}$ be the ring of Boolean functions on
$n$ variables $x_1, x_2, \ldots, x_n$, then
$$\text{B}_n=\text{F}_2[x_{1},\ldots,x_{n}]/(x_{1}^{2}+x_{1},\ldots,x_{n}^{2}+x_{n}),$$
and every $f\in \text{B}_n$ can be uniquely written in the
polynomial form $f=\sum_{I\in F_{2}^n}a_{I}x^I$, where
$x^I=x_{1}^{i_1}x_{2}^{i_2} \cdots x_{n}^{i_n}$, which is called the
algebraic normal form (ANF) of $f$. The algebraic degree of $f$,
denoted by $\deg(f)$, is the degree of the polynomial.

For $f \in \text{B}_{n}$, the algebraic immunity of $f$, denoted by
$\text{AI}(f)$ is defined to be the lowest degree of nonzero
annihilators of $f$ or $f+1$, i.e.,
\[\text{AI}(f)=\min\{\deg(g)~|~g\neq 0, fg=0\ {\rm or}~ (f+1)g=0\}.\]

Two Boolean functions $f$ and $g$ are called to be affine equivalent
if there exist $A\in \text{GL}_n(\text{F}_2)$ and $\varphi\in
\text{F}_2^n$ such that $g(x)=f(xA+\varphi)$. Clearly, algebraic
degree and algebraic immunity are both affine invariant.

Let $\alpha=(\alpha_1, \alpha_2, \ldots, \alpha_n)\in
\text{F}_{2}^n,$ the Hamming weight of $\alpha$, denoted by
$\wt(\alpha)$, is the number of $1$'s in $\{\alpha_1, \alpha_2,
\ldots, \alpha_n\}$. For an integer $i>0$ with 2-adic expansion
$i=\sum_{j=0}^mi_j2^j$, $\wt(i)$ represents the Hamming weight of
its binary expansion $(i_m, \ldots, i_1, i_0)_2$.

Let $\text{supp}(f)=\{x\in \text{F}_{2}^{n}~|~f(x)=1\},$ the
cardinality of $\text{supp}(f)$, denoted by $\wt(f)$, is called the
Hamming weight of $f$. We say that an $n$-variable Boolean function
$f$ is balanced if $\wt(f)=2^{n-1}$. The weight support
\cite{5:keqing Feng} of $f$, denoted by $\text{WS}(f)$, is defined
to be
\begin{equation*}\text{WS}(f)=\{i\mid \exists~x\in \text{supp}(f) ~{\rm
such~ that~} \text{wt}(x)=i\}.
\end{equation*}
We will use $P_{b}$ to represent the polynomial
$$(x_{1}+x_{2})(x_{3}+x_{4})\cdots(x_{2b-1}+x_{2b}).$$
Note that $P_{b}$ is a $(2b)$-variable polynomial with
$\deg(P_{b})=b$ and $\text{WS}(P_{b})=\{b\}$.

A Boolean function $f$ is symmetric if its output is invariant under
any permutation of its input bits, i.e.,
$$f(x_1, x_2, \ldots,
x_n)=f(x_{\sigma(1)}, x_{\sigma(2)}, \ldots, x_{\sigma(n)})$$ for
any permutation $\sigma$ of $\{1, 2, \ldots, n\}$.

Let $\text{SB}_{n}$ be the ring of symmetric Boolean functions on
$n$ variables $x_1, x_2, \ldots, x_n$, then every $f\in \text{SB}_n$
can be represented by a vector
$$v_f = (v_f(0), v_f(1), \ldots, v_f(n))\in \text{F}_2^{n+1},$$ where the
component $v_f(i)$ represents the function value for vectors of
weight $i$. The vector $v_f$ is called the simplified value vector
(SVV) of $f$. If $f\in \text{SB}_n$ and $f'(x_1, \ldots,
x_n)=f(x_1+1, \ldots, x_n+1),$ then $f'\in \text{SB}_n$ is affine
equivalent to $f$, and $v_{f'}(i)=v_f(n-i)$, for any $0\leq i\leq
n.$

On the other hand, the ANF of $f$ can be written as
$$f(x_1,
x_2, \ldots, x_n) = \sum_{i=0}^{n}\lambda_f(i)\sigma_i^n,$$ where
$\sigma_i^n$ is the homogeneous symmetric Boolean function on n
variables which consists of all the terms of degree $i$. The vector
$$\lambda_f = (\lambda(0), \lambda(1), \ldots, \lambda(n))\in
\text{F}_2^{n+1}$$ is called the simplified algebraic normal form
(SANF) vector of $f$. Both $v_f$ and $\lambda_f$ can be regarded as
mappings from $\{0, 1, \ldots, n\}$ to $\text{F}_2$.

\newtheorem{lem}{Lemma}[section]
\newtheorem{theo}{Theorem}[section]
\newtheorem{cor}{Corollary}[section]
\newtheorem{con}{Conjecture}[section]
\newtheorem{defi}{Definition}[section]
\newtheorem{construction}{Construction}[section]
\newtheorem{property}{Property}[section]

Let $a$ and $b$ be two nonnegative integers with 2-adic expansions
\[a = \sum_{j=0}^{m}{a_j2^j},~b = \sum_{j=0}^{m}{b_j2^j}.\] We say $a
\preceq b$ if $a_j \leq b_j$ for any $0\leq j\leq m$ and $a\prec b$
if $a\preceq b$ and $a\neq b$. Using the Lucas formula which states
that ${a \choose b} = 1\in \text{F}_2$ if and only if $b \preceq a$,
we can derive the following two lemmas:\vspace{0.3cm}
\begin{lem}\cite{:Canteaut } Let $f$ be an $n$-variable
symmetric Boolean function. Then $$\lambda_f(i) = \sum_{k\preceq i}
v_f(k), v_f(i) = \sum_{k\preceq i} \lambda_f(k).$$
\end{lem}\vspace{0.3cm}
\begin{lem}\cite{:Canteaut } For $\ell\geq 1$, suppose $f\in
\text{SB}_n$ and $\deg(f)< 2^\ell$, then $v_f$ has period $2^\ell$,
which means $v_f(i) = v_f(i+2^\ell)$ for $0\leq i \leq n-2^\ell$.
\end{lem}\vspace{0.3cm}

Lemma 2.2 can be derived easily from Lemma 2.1. \vspace{0.3cm}
\begin{lem}\cite{:D.K.Dalai} Let $n=2k$ and
$G_{n}$ be an $n$-variable symmetric Boolean function. If its
simplified value vector $v_{G_{n}}$ satisfies
\begin{equation*}v_{G_{n}}(i) =\begin{cases}0,\quad &{\rm for~}i\leq k,\\
  1,\quad &{\rm for~}i>k,
 \end{cases}\end{equation*}
then $\text{AI}(G_{n})=k$. Function $G_{n}$ is called the majority
function.
\end{lem}\vspace{0.3cm}
\begin{lem}\cite{5:keqing Feng} Suppose that $n\geq 2$ and $f\in \text{SB}_n$. If there exists
$0\neq g\in \text{B}_n$, such that $fg=0$, then there exists $b$,
$0\leq b\leq \lfloor \frac{n}{2}\rfloor$ and $0\neq h(x_{2b+1},
\ldots, x_n )\in \text{SB}_{n-2b}$, $\deg(h)\leq \deg(g)-b$, such
that $fhP_b=0$.
\end{lem}\vspace{0.3cm}

\begin{lem} Suppose $n = 2k$ and $f \in \text{SB}_n$.
If $\text{AI}(f) = k$, then $\wt(v_f) \in \{k, k + 1\}$.
\label{lem_suff_AIk}
\end{lem}
\begin{proof} It is sufficient to prove
that when $\wt(v_f) < k$ or $\wt(f) > k+1$, $f$ or $f+1$ has a
nonzero symmetric annihilator with degree less than $k$. Without
loss of generality, we consider that $\wt(v_f) < k$. Otherwise, we
can replace $f$ by $f+1$.

Let $g=\sum_{i=0}^{k-1}\lambda_g(i)\sigma_i$ be a symmetric
annihilator of $f$. Hence, $fg=0$ if and only if for all $v_f(i) =
1$, $v_g(i) = 0$ holds. Thus, by Lemma 2.1, $\wt(v_f)$ equations on
$k$ variables $\lambda_g(0), \ldots, \lambda_g(k-1)$ are obtained,
where the number of equations is less than the number of unknowns.
Therefore, at least one nonzero solution exists, which implies the
existence of such an annihilator.
\end{proof}

\section{ Necessary Conditions for Even-variable Symmetric Boolean Functions with Maximum \text{AI}}
We always assume $n=2k$. In this section, we will present the
constraints on the simplified value vector for an $n$-variable
symmetric Boolean function $f$ with maximum algebraic immunity $k$
step by step. First, we present Lemma 3.1 and Theorem 3.1, where
Lemma 3.1 is a special case of Theorem 3.1. According to Lemma 3.1
and Theorem 3.1, several notations and definitions are given. Based
on them, we present Corollary 3.1, Theorem 3.2, Theorem 3.3, and
Theorem 3.4, which are the main results of this section. Theorem 3.5
concludes this section by showing two classes of symmetric Boolean
functions satisfying all the necessary conditions. The following
lemma is very important.

\vspace{0.3cm}
\begin{lem}
\label{lem_antij} Let $n=2^{p+1}\mu$ with $p, \mu\geq 1$, and $f \in
\text{SB}_n$. If $\text{AI}(f) =k$, then
\begin{equation}
\label{equ_cons} v_{f}(2^{p}\mu- 2^{p} i+2^{p-1})=v_{f}(2^{p}\mu+
2^{p}j-2^{p-1})+1
\end{equation}
for any $1 \leq i, j \leq \mu$.
\end{lem}

\begin{proof} We will prove this theorem by induction on parameter $\mu$.

\textbf{Basis:} When $\mu = 1$, it is true due to Theorem 2.2 of
$\cite{2:Qu}$.

\textbf{Induction:} Assuming the theorem is true for $\mu=\ell\geq
1$, we claim that it is also true for $\mu=\ell+1$. Now, let
$n=2k=2^{p+1}(\ell+1)$ and $f \in \text{SB}_{n}$ with
$\text{AI}(f)=k=2^p(\ell+1)$. We will prove that the Boolean
function $f$ satisfies
$v_{f}(2^{p}(\ell+1)-2^{p}i+2^{p-1})=v_{f}(2^{p}(\ell+1)+2^{p}j-2^{p-1})+1$
for any $1\leq i,j \leq \ell+1$ in the following four steps.

In step 1, we prove \eqref{equ_cons} for $1\leq i,j\leq \ell$; in
step 2, for $1\leq j\leq \ell$ and $i=\ell+1$; in step 3, for $1\leq
i\leq \ell$ and $j=\ell+1$; and in step 4, for $i=\ell+1$ and
$j=\ell+1$.

\vspace{0.3cm} \textbf{Step 1}  Assume to the contrary that
$v_{f}(2^{p}(\ell+1)- 2^{p} i+2^{p-1})=v_{f}(2^{p}(\ell+1)+
2^{p}j-2^{p-1})$ for some $1\leq i, j\leq \ell$, letting $f'\in
\text{SB}_{2^{p+1}\ell}$ on variables $x_{2^{p+1}+1}$,
$x_{2^{p+1}+2}$, $\ldots$, $x_{2^{p+1}(\ell+1)}$, be defined as
$$
v_{f'}=(v_f(2^p), v_f(2^p + 1), \ldots, v_f(2^{p+1}\ell + 2^p)).
$$
Then we have $v_{f'}(2^{p}\ell-
2^{p}i+2^{p-1})=v_{f'}(2^{p}\ell+2^{p}j-2^{p-1})$. By induction
hypothesis, there exists $0\neq h\in B_{2^{p+1}\ell}$ on variables
$x_{2^{p+1}+1}$, $x_{2^{p+1}+2}$, $\ldots$, $x_{2^{p+1}(\ell+1)}$
with $\deg(h) < 2^p \ell$, such that $hf' = 0$ or $h(f'+1)=0$.

For the case $hf'=0$, let $g = hP_{2^p} $. Then, we claim that $f g
= 0$. To prove the foregoing, we study the weight supports of $g$
and $f$. First, by $\text{WS}(P_{2^{p}})=\{2^{p}\}$ and the fact
that $P_{2^{p}}$ and $h$ deal with different variables, we know that
$WS(g) = \{i+2^p|i\in \text{WS}(h)\}$ and $\text{WS}(g)\cap\{i|0\leq
i< 2^{p}, n-2^p<i\leq n\}= \emptyset $. Second, we know that
$\text{WS}(f)=\{i+2^p|i\in \text{WS}(f')\}\cup\{i|0\leq i< 2^{p},
n-2^p<i\leq n, v_f(i)=1\} $ by the definition of $f'$. Third, we
have $\text{WS}(f')\cap \text{WS}(h)=\emptyset$ because $f' h=0$.
Thus, we have $\text{WS}(f)\cap \text{WS}(g)=\{i+2^{p}|i\in
\text{WS}(f')\}\cap \{i+2^{p}|i\in \text{WS}(h)\}=\emptyset$, which
means $f g = 0$. For the case $h(f'+1)=0$, we can prove similarly
that $(f+1)g=0$. This contradicts $AI( f ) = k$ because $\deg(g) <
k$. Therefore,
\begin{equation}
\label{equ_step1} v_{f}(2^{p}(\ell+1)- 2^{p}
i+2^{p-1})=v_{f}(2^{p}(\ell+1)+ 2^{p}j-2^{p-1})+1
\end{equation}
for any $1\leq i, j\leq \ell$.

\vspace{0.3cm} \textbf{Step 2} Assume to the contrary that
$v_{f}(2^{p-1})=v_{f}(2^{p}(\ell+1)+2^{p}j-2^{p-1})$ for $1\leq
j\leq \ell$. To deduce a contradiction, we construct an annihilator
$g$ of $f$ or $f+1$ as follows.

Define $g\in \text{SB}_{2^{p+1}(\ell+1)}$ by
\begin{equation*}\lambda_g(\psi)=
\begin{cases}
0,\hspace{0.1cm}\text{if }\psi \ge k\text{ or }2^{p-1} \not\preceq \psi,\\
1,\hspace{0.1cm}\text{if }\psi < k\text{ and }2^{p-1} \preceq \psi.\\
\end{cases}
\end{equation*}
We claim that $g$ is an annihilator of $f$ or $f+1$. To prove the
claim, we study the weight support of $g$. Let $\omega\in
\text{WS}(g)$, then $v_g(\omega)=1$. By Lemma 2.1, we have
$$v_g(\omega)=\sum_{\psi\preceq \omega}\lambda_g(\psi)=\sum_{\begin{array}{c}
                                                                2^{p-1}\preceq\psi\preceq\omega, \\
                                                                  \psi<k
                                                               \end{array}
}1.$$ Let $$S_{\omega}=\{\psi | 2^{p-1}\preceq\psi\preceq
\omega,\psi<k=2^p(l+1)\},$$ then $v_g(\omega)=|S_{\omega}| \mod 2$,
which means $v_g(\omega)=1$ if and only if $|S_{\omega}|$ is odd.
Let $\omega=(\omega_{m}\omega_{m-1}\cdots\omega_0)_2$ be the binary
expansion of $\omega$.
\begin{itemize}
\item [i)] If $\omega<k$, we claim that $|S_{\omega}|$ is odd if and only if
$\omega=2^{p-1}$. For $2^{p-1}\not\preceq \omega$, there is no
$\psi$ satisfying $2^{p-1}\preceq\psi\preceq\omega$, which means
$S_{\omega}=\emptyset$. For $2^{p-1}\preceq \omega$, the number of
$\psi$ such that $2^{p-1}\preceq\psi\preceq \omega$ is
$2^{\text{wt}(\omega)-\text{wt}(2^{p-1})}=2^{\text{wt}(\omega)-1}$,
which is odd only when $\omega=2^{p-1}$.
\item [ii)] If $\omega\geq k$, we claim that $|S_{\omega}|$ is
odd only if $\omega_{p-1}=1$ and $\omega_t=0$ for all $0\leq t<p-1$.
Otherwise, if $\omega_{p-1}=0$, then $2^{p-1}\not\preceq\omega$ and
there is no $\psi$, such that $2^{p-1}\preceq\psi\preceq\omega$,
which implies that $S_\omega=\emptyset$. If $\omega_{t_0}=1$ for
some $0\leq t_0<p-1$, and $\psi=(\psi_{m}\psi_{m-1}\cdots\psi_0)_2$
is an element of $ S_{\omega}$, then it is clear that
$\psi'=(\psi_{m}\cdots\psi_{t_0+1}\overline{\psi_{t_0}}\psi_{t_0-1}\cdots\psi_0)_2$
also satisfies $2^{p-1}\preceq \psi'\preceq\omega$ and
$\psi'<k=2^{p}(\ell+1)$, where $\overline{\psi_{t_0}}=\psi_{t_0}+1$.
Thus, $\psi'$ is also an element of $S_{\omega}$, which means the
elements in $S_{\omega}$ come into pairs. Thus, $|S_{\omega}|$ is
even, which is a contradiction. Therefore, if $v_{g}(\omega)=1$,
then $\omega_{p-1}=1$ and $\omega_t=0$ for all $0\leq t<p-1$.
\end{itemize}
Combining the results of i) and ii), we have
$$\text{WS}(g)\subseteq \{2^{p-1},
2^p(\ell+1)+2^pj-2^{p-1},1\leq j\leq \ell+1\}.$$ Note that
$v_f(2^{p-1})=v_f(2^p(\ell+1)+2^pj-2^{p-1}), 1\leq j\leq \ell$, if
we can prove that $2^p(\ell+1)+2^p(\ell+1)-2^{p-1}=n-2^{p-1}$ is not
in $\text{WS}(g)$, then we have $fg=0$ or $(1+f)g=0$ since $f$ is
constant on the support of $g$. Since $\deg(g)<k$, we have
$AI(f)<k$. It is a contradiction, and will end the proof of this
part. Therefore, we will prove $v_g(n-2^{p-1})=0$. Note that once it
is proved, we finish the proof of this part.

\par Let $\omega=n-2^{p-1}=2^{p+1}\ell+2^p+2^{p-1}$, $\psi$ be an element of
$S_\omega$. According to the definition of $S_{\omega}$, we can see
that there exists some integer $0\leq s\leq \ell$, such that
$\psi=2^ps+2^{p-1}$. Let $T_\ell=\{s|s\preceq 2\ell+1, 0\leq s\leq
\ell\}$. Hence we have $|S_\omega|=|T_\ell|$ by the definition of
$S_{\omega}$. What we need is to prove that $|T_\ell|$ is even for
all $\ell\geq 1$. It is not a difficult task, and the reader can
give a proof by himself/herself, or follow the proof below.

\par If $\ell=2^r-1$ for some positive integer $r$, then
$2\ell+1=2^{r+1}-1$. Thus, $s\preceq 2\ell+1$ for every $0\leq s\leq
\ell$, which means $|T_\ell|=\ell+1=2^r$. It is in contradiction
with that $|S_{\omega}|=|T_{\ell}|$ is odd. Otherwise, let
$\ell=(\ell_m\ell_{m-1}\cdots \ell_0)_2$ be the binary expansion of
$\ell$, then there exists some integer $1\leq t\leq m$ such that
$\ell_m=\ell_{m-1}=\cdots=\ell_{m-(t-1)}=1$ and $\ell_{m-t}=0$,
namely $\ell=( \underbrace{11\cdots 1}_{t}0\ell_{m-t-1}\cdots
\ell_0)_2$. Then $2\ell+1=( \underbrace{11\cdots
1}_{t}0\ell_{m-t-1}\cdots \ell_01)_2$. Let $s=(s_ms_{m-1}\cdots
s_0)_2$, then $s\preceq 2\ell+1$ implies $s_{m-t+1}=0$, which means
$s<\ell$ by the structure of $\ell$. Thus, by the definition of
$T_{\ell}$, $s\in T_{\ell}$ if and only if $s\preceq2\ell+1$. Since
$s<2^{m+1}$ and $(2\ell+1)_{m+1}=1$, where $(2\ell+1)_{m+1}$ denotes
the $(m+1)^{\text{th}}$ bit in its binary expansion, we have
$|T_{\ell}|=2^{\text{wt}(2\ell+1)-1}$. Since $\ell\geq 1$, we have
$\text{wt}(2\ell+1)-1> 0$ which means $|S_{\omega}|=|T_{\ell}|$ also
even. Thus, we finish the proof of this step, i.e.,
\begin{equation}\label{equ_step2}
v_{f}(2^{p-1})=v_{f}(2^{p}(\ell+1)+2^{p}j-2^{p-1})+1\end{equation}
for $1\leq j\leq \ell$.

\vspace{0.3cm} \textbf{Step 3}  Assume to the contrary that
$v_{f}(2^{p}(\ell+1)-2^{p}i+2^{p-1})=v_{f}(2^{p+1}(\ell+1)-2^{p-1})$,
for $1\leq i\leq \ell$. Similar with step 2, by using $g'$ instead
of $g$, where $g'(x_1, \ldots, x_n)=g(x_1+1, \ldots, x_n+1),$ we can
get $fg'=0$ or $(f+1)g'=0$, which contradicts $\text{AI}(f)=k$.
Thus,
\begin{equation}
v_{f}(2^{p}(\ell+1)-2^{p}i+2^{p-1}) =
v_{f}(2^{p+1}(\ell+1)-2^{p-1})+1 \label{equ_step3}
\end{equation}
for any $1\leq i\leq \ell$.


\vspace{0.3cm} \textbf{Step 4} Combining the above three steps, we
have
\begin{equation*}
\begin{array}{lll}
&v_{f}(2^{p-1})&\\
=&v_{f}(2^{p}(\ell+1)+2^{p}j-2^{p-1})+1 & \text{ by \eqref{equ_step2}}\\
=&v_{f}(2^{p}(\ell+1)-2^{p}i+2^{p-1}) & \text{ by \eqref{equ_step1}} \\
=&v_{f}(2^{p+1}(\ell+1)-2^{p-1})+1 & \text{ by \eqref{equ_step3}}\\
\end{array}
\end{equation*}
for any $1\leq i, j \leq \ell$. Thus,
$v_{f}(2^{p}(\ell+1)-2^{p}i+2^{p-1})=v_{f}(2^{p}(\ell+1)+2^{p}j-2^{p-1})+1$
for $i=j= \ell+1$. \vspace{0.2cm}

Combining the above four steps, $v_{f}(2^{p}(\ell+1)-
2^{p}i+2^{p-1})=v_{f}(2^{p}(\ell+1)+ 2^{p}j-2^{p-1})+1$ holds for
any $1\leq i, j \leq \ell+1$. Therefore, the theorem is also true
for $\mu=\ell+1$. This completes the proof.
\end{proof}

\vspace{0.3cm} In Lemma 3.1, $n$ should be a multiple of 4. The
following theorem generalizes Lemma \ref{lem_antij} to a wider
situation, where $n$ can be any even number.

\vspace{0.3cm}
\begin{theo}
\label{thm_antij} Let $n=2^{p+1}\mu+2m$ with $p, \mu \geq 1$ and
$0\leq m< 2^{p}$, $f\in \text{SB}_{n}$. If $\text{AI}(f)=k$, then
$$v_{f}(2^{p}\mu+m-2^{p}i+2^{p-1})=v_{f}(2^{p}\mu+m+2^{p}j-2^{p-1})+1$$
for any $1\leq i,j\leq \mu$.
\end{theo}
\begin{proof}
Assume to the contrary that
$v_{f}(2^{p}\mu+m-2^{p}i+2^{p-1})=v_{f}(2^{p}\mu+m+2^{p}j-2^{p-1})$
for some $1\leq i, j\leq \mu$. Let $f'\in \text{SB}_{2^{p+1}\mu}$ on
variables $x_{2m+1}$, $x_{2m+2}$, $\ldots$, $x_{2^{p+1}\mu+2m}$, be
defined as
$$
v_{f'} = (v_f(m), v_f(m + 1), \ldots, v_f(2^{p+1}\mu+ m)).
$$
Then we have
$v_{f'}(2^{p}\mu-2^{p}i+2^{p-1})=v_{f'}(2^{p}\mu+2^{p}j-2^{p-1})$.
By Lemma \ref{lem_antij}, $\text{AI}(f')<2^p \mu$, thus there exists
$0\neq h\in B_{2^{p+1}\mu}$ with $\deg(h) < 2^p \mu$ such that $f'h
= 0$ or $(f'+1)h=0$. Let $g=hP_{m}$. Following the argument of step
1 in Lemma \ref{lem_antij}, we claim that $fg = 0$ or $(f+1)g=0$
with $\deg(g)< 2^p \mu+m=k$, which contradicts $\text{AI}(f)=k$.
Therefore,
$v_{f}(2^{p}\mu+m-2^{p}i+2^{p-1})=v_{f}(2^{p}\mu+m+2^{p}j-2^{p-1})+1$
for any $1\leq i, j\leq \mu$.
\end{proof}\vspace{0.3cm}

For example, when $n=2k=14$, we have
\begin{itemize}
\item if $p=1$, $\mu=3$, $m=1$, then $\{2^{p}\mu+m-2^{p}i+2^{p-1}, 2^{p}\mu+m+2^{p}j-2^{p-1}|1\leq i,j\leq \mu\}=\{2,4,6,8,10,12\}$,
\item if $p=2$, $\mu=1$, $m=3$, then $\{2^{p}\mu+m-2^{p}i+2^{p-1}, 2^{p}\mu+m+2^{p}j-2^{p-1}|1\leq i,j\leq \mu\}=\{5,9\}$.
\end{itemize}
Theorem 3.1 sets constraints on $v_{f}(\omega)$ for
$\omega\in\{2,4,5,6,8,9,$ $10,$ $12\}$.

For convenience of description, we introduce a partial order on
nonnegative integers denoted as $\preceq'$.


\vspace{0.3cm}\begin{defi} Given two binary expansions of
nonnegative integers $a=(a_s, a_{s-1}, \ldots, a_0)_2$, $b=(b_\ell,
b_{\ell-1}, \ldots, b_0)_2$, $1\leq \ell\leq s$, we define
\begin{itemize}
\item[ ]$\hspace{1cm}$ $b\preceq' a \Leftrightarrow b=0 ~\text{or} ~b_{i}=a_{i}$ \text{ for all } $0\leq i\leq
\ell$;
\item[ ]$\hspace{1cm}$ $b\prec' a \Leftrightarrow b\preceq'a$ and $b\neq
a$.
\end{itemize}
\end{defi}\vspace{0.3cm}
For example, we have $3\prec'7$ because $7=(111)_{2}$ and
$3=(11)_{2}$.

For any nonnegative integer $k$, let $\text{B}^k=\{i,
2k-i~|~i\prec'k\}$. By the definition of $\prec'$,
$|\text{B}^k|=2\wt(k)$. \vspace{0.3cm}
\begin{defi}
\label{def_Atk} For any positive integer $n=2k$, we divide the set
$\{0, 1, \ldots, n\}$ into a series of subsets $\text{A}_i^k$, where
\begin{itemize}
\item[ ]$\text{A}_0^k=\{k\}$,
\item[ ]$\text{A}_i^k=\{k-(2j+1)2^{i-1},
k+(2j+1)2^{i-1}~|~0\leq j\leq
\frac{\lfloor\frac{k}{2^{i-1}}\rfloor-1}{2}\}$,
\end{itemize}
for $1\leq i\leq \lfloor \log_2 n \rfloor$. The superscript $k$ may
be omitted if there is no confusion.
\end{defi}\vspace{0.3cm}
The union of sets $\{(2j+1)2^{i-1}\mid j\in N\}$ over all $i\in
N^{+}$ is a partition of $N^{+}$, so these subsets have no
intersection with each other, and $\{0, 1, \ldots,
n\}=\bigcup_{i=0}^{\lfloor \log_2 n \rfloor}\text{A}_i$.

For example, when $n=2k=14$, we have $\text{A}_{0}^7=\{7\}$,
$\text{A}_{1}^7=\{0,2,4,6,8,10,12,14\}$,
 $\text{A}_{2}^7=\{1,5,9,13\}$,
 $\text{A}_{3}^7=\{3,11\}$.
\par

The main intuition of sets $A^{k}_{i}$ and $\prec'$ could be
explained by the binary expansion of $k$, where
$k=(k_{m},...,k_{0})_2$ with $k_{m}\not=0$. For any $a \prec' k$, it
is easy to verify that $a=(\overline{k_{j}},k_{j-1},...,k_{0})_{2}$,
where $k_{j}=1$ and $\overline{k_{j}}=k_{j}+1=0$. And for every
$\omega\in A^{k}_{i}$, the binary expansion of $\omega$ is
$(*,\overline{k_{i-1}},k_{i-2},...,k_{0})_{2}$, where $*$ is an
arbitrary binary string, which means the right-most $i$ bits of the
binary expansion of $\omega$ are exactly
$\overline{k_{i-1}},k_{i-2},...,k_{0}$.



The following Lemma contains some properties of $A_i^k$, as well as
the partial order $\prec'$.

\vspace{0.3cm}
\begin{lem}
Supposing $k=(k_m, \ldots, k_1, k_0)_2$, by the definition of
$\text{A}_{i}^k$, a simple calculation gives the following:

1) if $j\in \text{A}_i^k$, then $2k-j\in \text{A}_i^k$;

2) for any $0\leq j\leq 2k$, $j=(j_m,$ $\ldots,$ $j_1,$ $j_0)_2$,
$j\in \text{A}_i^k$ if and only if $(j_{i-1},$ $j_{i-2},$ $\ldots,$
$j_1,$ $j_0)$=$(\overline{k_{i-1}},$ $k_{i-2},$ $\ldots,$ $k_1,$
$k_0)$, in particular, $\text{A}_0^k=\{k\}$, $\text{A}_{\lfloor
\log_2k \rfloor+1}^k$=$\{k-2^{\lfloor \log_2k \rfloor},$
$k+2^{\lfloor \log_2k \rfloor}\}$;

3) $\text{A}_i^k$ contains an element $(\overline{k_{i-1}},k_{i-2},$
$\ldots,$ $k_1,$ $k_0)_2\prec'k$ in $B^k$ if and only if
$k_{i-1}=1$.\vspace{0.3cm}

We explain the reason why we define the sets $\text{A}^{k}_{i}$ and
$\text{B}^{k}$. Given $n=2k=2^{p+1}\mu+2m$ ($p,\mu\geq1$, $0\leq
m<2^{p}$), it is easy to verify that
$$\text{A}^{k}_{p}\supseteq \{2^{p}\mu+m-2^{p}i+2^{p-1},2^{p}\mu+m+2^{p}j-2^{p-1}
|1\leq i,j\leq \mu\},$$ which means the $\omega$ of $v_{f}(\omega)$
from the same equation defined by Theorem 3.1 are all included in
the same $\text{A}^{k}_{p}$. But $\text{A}^{k}_{p}$ may contain two
extra elements, which are
\begin{equation*}
\text{A}^{k}_{p}-\{2^{p}\mu+m-2^{p}i+2^{p-1},2^{p}\mu+m+2^{p}j-2^{p-1}
|1\leq i,j\leq \mu\}
\end{equation*}
\begin{equation*}=
\begin{cases}
\{m-2^{p-1}, n-m+2^{p-1}\},& m-2^{p-1}\geq 0,\\
\emptyset,& m-2^{p-1}<0.
\end{cases}
\end{equation*}

For the case $m-2^{p-1}\geq 0$, since $m<2^{p}$, we have
$m=2^{p-1}+s$ ($0\leq s<2^{p-1}$), which means
$m-2^{p-1}\prec'm\prec'k$. Besides, by the definition of partial
order $\prec'$, if $m+2^{p}\mu=k$ ($m<2^{p}$, $\mu>1$), then
$2^{p}\mu+m-2^{p}i+2^{p-1}\not\preceq'k$ for $1\leq i\leq \mu$.
Thus, Theorem 3.1 shows constraints on $v_{f}(\omega)$ if and only
if $\omega\not\preceq' k$, namely $\omega\in
\{0,1,...,n\}-\text{B}^{k}-\{k\}$.
\end{lem}\vspace{0.3cm}

Equipped with these notations and basic properties, we can restate
Theorem \ref{thm_antij} concisely.

\vspace{0.3cm}
\begin{cor}
\label{cor_antij}
 Let $n=2k$ and $f\in \text{SB}_n$. If $\text{AI}(f)=k$, then for any
$1\leq p\leq \lfloor \log_2 k \rfloor$ and $i, j\in
\text{A}_p^k-\text{B}^k-\{k\}$ with $i\leq j<k$, we have
$v_f(i)$=$v_f(j)$=$v_f(n-j)+1$=$v_f(n-i)+1$.
\end{cor}\vspace{0.3cm}

Corollary \ref{cor_antij} shows the constraints related to the
values of $v_{f}$ on $\{0, 1, \ldots, n\}-\text{B}^{k}-\{k\}$. In
what follows, we will discuss $v_{f}$ on $B^{k}\cup\{k\}$.

\vspace{0.3cm}
\begin{theo}
\label{thm_converse} Let $f\in \text{SB}_{n}$. For any $t\prec'k$,
$t\neq k-2^{\lfloor \log_2k \rfloor}$, assume $t\in \text{A}_p$. If
\begin{eqnarray}
\label{noequal} v_{f}(t)+1 & = & v_{f}(t + 2^p)= \cdots = v_{f}(k -
2^{p-1})\nonumber
\\ & = & v_{f}(k+2^{p-1})+1=\cdots=v_{f}(n-t-2^p)+1 \nonumber
\\ & = & v_{f}(n-t),
\end{eqnarray}
then $\text{AI}(f)<k$.
\end{theo}
\begin{proof}
Notice that $\text{A}_p=$ $\{t, t+2^p,$ $\ldots,$ $k-2^{p-1},$
$k+2^{p-1},$ $\ldots,$ $n-t-2^{p-1},$ $n-t\}$. If $k=2^{q}$ for some
$q$, then only one $i$ ($i=0$) satisfies $i\prec'k$ and
$0=k-2^{\lfloor \log_2k \rfloor}$. It contradicts the conditions in
this theorem. Therefore, we have $k\neq 2^q$ for any integer $q$. We
only need to consider $\wt(v_f)=k$ or $k+1$. Otherwise, we have
$AI(f)<k$ by Lemma \ref{lem_suff_AIk}. Without loss of generality,
we assume $\wt(v_f) = k$. Otherwise when $\wt(v_f) = k+1$, we can
replace $f$ by $f+1$ instead.

\par  We will prove that
there exists a nonzero symmetric Boolean function $g$ with degree
less than $k$, such that $fg=0$ which implies $AI(f)<0$. Let
$g=\sum_{i=0}^{k-1}\lambda_g(i)\sigma_i$. Notice that $fg=0$ if and
only if for every $v_f(i) = 1$ we have
$$
v_g(i)=\sum_{\substack{j \preceq i\\0 \leq j \leq k -
1}}{\lambda_g(j)}=0
$$by Lemma 2.1. Then, we can get a system of homogeneous linear equations on
variables $\lambda_g(0), \ldots, \lambda_g(k-1)$ with $wt(v_{f})=k$
equations. The number of equations and unknowns of the equation
system are both $k$. In what follows, we will show that there are
two same equations. Thus there must exist a nonzero solution of
$\lambda_g(0), \ldots, \lambda_g(k-1)$,
 which implies the existence of $g$.




Since $k\not=2^{q}$ for any integer $q$, we assume $2^{\ell-1}< k <
2^\ell$. Thus, we have $t<2^{\ell-1}$ and
$\lfloor\log_{2}k\rfloor=\ell-1$. Since $t\not=
k-2^{\lfloor\log_{2}k\rfloor}$, we have $2t\not= n-2^{\ell}$
$\Rightarrow$ $n-t\not= t+2^{\ell}$ $\Rightarrow$ $n-t-2^{\ell}\not=
t$. According to the definition of $A_{p}$, we have $t + 2^\ell, n -
t - 2^\ell\in \text{A}_p$.

For the case $v_f(t) = 1$, since $t+2^{\ell}\in A_{p}$, $t+2^\ell>k$
and $n-t\neq t+2^\ell$, we have $v_f(t+2^\ell)=v_{f}(t) = 1$ by
\eqref{noequal}. Consider the equations
\begin{equation*}v_g(t)=\sum_{\substack{i\preceq t\\
0\leq i<k}}\lambda_{g}(i)=0
\end{equation*}
and
$$\hspace{0.3cm} v_g(t + 2^\ell)=\sum_{\substack{i\preceq t+2^\ell\\
0\leq i<k}}\lambda_{g}(i)=0.$$
It is easy to see that $i\preceq t$ is equivalent to $i\preceq t +
2^\ell$ for $0 \leq i<k$; thus, the two equations above are exactly
the same.

For the case $v_f(t)=0$, we could prove $v_f(n-t)=v_f(n-t-2^\ell)=1$
similar to the case $v_{f}(t)=1$. It is similar to verify that
equations $v_g(n-t)=0$ and $v_g(n-t-2^\ell)=0$ are exactly the same.

Therefore, the nonzero symmetric annihilator with degree less than
$k$ always exists, and $\text{AI}(f)<k$.
\end{proof}

\vspace{0.3cm}
For a given $k$, the values $t$ and $n-t$ such that $t\prec'k$ and
$t\not=k-2^{\lfloor\log_{2}k\rfloor}$ can occur in Theorem 3.2, but
are excluded in Corollary 3.1. Theorem 3.2 focuses on the
relationship between $v_{f}(\omega)$ where $\omega\in\{t,n-t\}$, and
$v_{f}(\omega)$ where $\omega\in A_{p}-\{t,n-t\}$. In the following
theorem, we will consider $v_{f}(k-2^{\lfloor\log_{2}k\rfloor})$ and
$v_{f}(n-k+2^{\lfloor\log_{2}k\rfloor})$, which are excluded in
Theorem 3.2.

\vspace{0.3cm}
\begin{theo}
\label{thm_atmostone} Let $n=2k$ and $f\in \text{SB}_{n}$. If
$AI(f)=k$, then there does \emph{not} exist more than one integer
$i$, such that $i\prec'k$ and $v_{f}(i)=v_{f}(n-i)$.
\end{theo}

\begin{proof}
When $k = 2^q$ and $q$ is an integer, there is only one $i (i=0)$
satisfying $i\prec'k$. The conclusion is trivial. Therefore, we only
need to consider the case $k\not=2^{q}$ for any integer $q$. By
Lemma \ref{lem_suff_AIk}, we only need to consider $\wt(v_f)=k$ or
$k+1$. Without loss of generality, we assume $\wt(v_f) = k$.
Otherwise, when $\wt(v_f) = k+1$, we can replace $f$ by $f+1$
instead.

Assume to the contrary that there exist more than one $i$ such that
$i\prec'k$ and $v_{f}(i)=v_{f}(n-i)$. We will show the existence of
a nonzero symmetric Boolean function $g$ with degree less than $k$
such that $fg=0$, which is contradicted with $AI(f) = k$.

Let $g=\sum_{i=0}^{k-1}\lambda_g(i)\sigma_i$. Notice that $fg=0$ if
and only if for every $v_f(i) = 1$ we have
$$
v_g(i)=\sum_{\substack{j \preceq i\\0 \leq j \leq k -
1}}{\lambda_g(j)}=0
$$by Lemma 2.1. Then, we can get a system of homogeneous linear equations on
variables $\lambda_g(0), \ldots, \lambda_g(k-1)$ with $wt(v_{f})=k$
equations. Notice the number of equations and unknowns are both $k$.
In what follows, we will show that there are two same equations;
thus, there must exist a nonzero solution of $\lambda_g(0), \ldots,
\lambda_g(k-1)$, which implies the existence of $g$.


We claim that there exists at least one $i_{1}$ such that
$i_{1}\prec'k$ and $v_{f}(i_{1})=v_{f}(n-i_{1})=1$ under our
assumption that more than one $i \prec' k$ exist s.t.
$v_{f}(i)=v_{f}(n-i)$. Else, suppose $v_{f}(i)=v_{f}(n-i)=0$ for all
$i\prec' k$ and $v_{f}(i)=v_{f}(n-i)$. It is easy to verify that
$\text{wt}(f)<k$, because $v_{f}(i)=v_{f}(n-i)+1$ for other
$i\prec'k$ and $v_{f}(\psi)=v_{f}(n-\psi)+1$ for all $\psi \in
\{0,1,...,n\}-B^{k}-\{k\}$ due to Corollary 3.1. This is
contradicted with $\text{wt}(f)=k$. Thus the existence of $i_{1}$ is
guaranteed. Since $k\not=2^{q}$ for any $q$, we assume
$2^{\ell-1}<k<2^{\ell}$.


\vspace{0.3cm} \textbf{Case 1:} If
$i_{1}\not=k-2^{\lfloor\log_{2}k\rfloor}$, assume $i_{1}\in
A^{k}_{p}$. Since $AI(f)=k$, according to Corollary 3.1, we have
\begin{eqnarray}
v_{f}(i_{1}) & = & v_{f}(i_{1} + 2^p)= \cdots = v_{f}(k -
2^{p-1})\nonumber
\\ & = & v_{f}(k+2^{p-1})+1=\cdots=v_{f}(n-i_{1}-2^p)+1 \nonumber
\\ & = & v_{f}(n-i_{1})=1, \label{case_11}
\end{eqnarray}
or
\begin{eqnarray}
v_{f}(i_{1}) & = & v_{f}(i_{1} + 2^p)+1= \cdots = v_{f}(k -
2^{p-1})+1\nonumber
\\ & = & v_{f}(k+2^{p-1})=\cdots=v_{f}(n-i_{1}-2^p) \nonumber
\\ & = & v_{f}(n-i_{1})=1. \label{case_12}
\end{eqnarray}
By the definition of $A^{k}_{p}$, we have $i_{1}+2^{\ell}$,
$n-i_{1}-2^{\ell}\in A^{k}_{p}$. Since $i_{1}+2^{\ell}>k$ and
$n-i_{1}-2^{\ell}<k$, we have
$v_{f}(n-i_{1})=v_{f}(n-i_{1}-2^{\ell})=1$ for \eqref{case_11} and
$v_{f}(i_{1})=v_{f}(i_{1}+2^{\ell})=1$ for \eqref{case_12}. Then,
similar with the proof in Theorem 3.2, we can prove that they are
two same equations in both cases, i.e., $v_g(n-i_{1})=0$ is
equivalent to $v_g(n-i_{1}-2^\ell)=0$ and $v_g(i_{1})=0$ equivalent
to $v_g(i_{1}+2^\ell)=0$.

\vspace{0.3cm} \textbf{Case 2:} For
$i_{1}=k-2^{\lfloor\log_{2}k\rfloor}$, we have $i_{1}=k-2^{\ell-1}$
$\Rightarrow$ $2i_{1}=n-2^{\ell}$ $\Rightarrow$
$n-i_{1}=2^{\ell}+i_{1}$. Thus, we have
$v_{f}(i_{1})=v_{f}(n-i_{1})=v_{f}(i_{1}+2^{\ell})=1$. Consider the
following two equations,
\begin{eqnarray*}
v_{g}(n-i_{1})
&=&v_{g}(2^{\ell}+i_{1})\\
&=&\sum_{\scriptstyle i\preceq 2^{\ell}+i_{1}\atop\scriptstyle 0\le i<k}\lambda_{g}(i)=0,\\
\end{eqnarray*}
and
\begin{eqnarray*}
v_{g}(i_{1})&=&\sum_{\scriptstyle i\preceq i_{1}\atop\scriptstyle
0\leq i<k}\lambda_{g}(i)=0.  \\
\end{eqnarray*}

For $0\leq i<k$, we have $i\preceq2^{\ell}+i_{1}$ if and only if
$i\preceq i_{1}$. Thus, the above two equations are equivalent.



\par Therefore, the nonzero symmetric annihilator with degree less
than $k$ always exists, which is contradictory to $AI(f)=k$. Thus,
there cannot exist more than one integer $i$, such that $i\prec'k$
and $v_{f}(i)=v_{f}(n-i)$.
\end{proof}

\vspace{0.3cm} For the case one $t$ exists such that $t\prec'k$ and
$v_{f}(t)=v_{f}(n-t)$, there exists another constraint, namely
Theorem 3.4. This theorem is the last necessary condition for
even-variable symmetric Boolean functions to reach maximum algebraic
immunity, which considers all the triples
$(v_{f}(t),v_{f}(k),v_{f}(n-t))$ when $t\prec' k$.


\vspace{0.3cm}
\begin{theo}
\label{thm_vfk} Let $n=2k$, $f\in \text{SB}_{n}$. If
$\text{AI}(f)=k$, then for any $t\prec' k$,
$(v_{f}(t),v_{f}(k),v_{f}(n-t))\not\in\{(0,0,0),(1,1,1)\}$.
\end{theo}

\begin{proof}
According to Corollary 3.1, for any $i$ and any $p$ such that $i\in
A_{p}-B^{k}-\{k\}$, we have $v_{f}(i)=v_{f}(n-i)+1$.
 By Theorem 3.3, for all elements of $\text{B}^{k}$, at most one $t$ could exist such that $t\prec'k$ and $v_{f}(t)=v_{f}(n-t)$.
 Therefore, $\wt(v_{f}) = k-1$ if $(v_{f}(t),v_{f}(k),v_{f}(n-t)) = (0, 0, 0)$ for some $t \prec' k$ and $\wt(v_{f}) = k+2$ if $(v_{f}(t),v_{f}(k),v_{f}(n-t)) = (1, 1, 1)$ for some $t \prec' k$.
 By Lemma \ref{lem_suff_AIk}, we know either case is impossible.
\end{proof}
\vspace{0.3cm}

In the end of this section, we take out all even-variable symmetric
Boolean functions satisfying all the necessary conditions to achieve
maximum algebraic immunity into the following two classes.

\vspace{0.3cm}
\begin{defi}
\label{def_class} Define two classes of symmetric Boolean functions
on $n$ variable, $n = 2k$, as follows.
\begin{itemize}
\item[]\textbf{Class 1}: For any $\text{A}_p$, $1\leq p\leq\lfloor \log_{2}n\rfloor $, and $i,j\in \text{A}_p$,
\begin{equation*}v_{f}(i)=
\begin{cases}
v_{f}(j)+1,\hspace{0.1cm}{\rm~if~}\hspace{0.1cm}i<k<j\textrm{ or }j<k<i,\\
v_{f}(j),\hspace{0.6cm}{\rm~otherwise}.
\end{cases}
\end{equation*}

\item[]\textbf{Class 2}: There is a function $g$ contained in Class 1 and an integer $t\prec'
k$ such that
\begin{itemize}
\item[]\hspace{1.7cm} $v_f=v_g+e_{t}+\delta e_{k}$, or
\item[]\hspace{1.7cm} $v_f=v_g+e_{n-t}+\delta'e_{k}$,
\end{itemize}
where
\begin{equation*}\delta=v_g(t)+v_g(k),
\end{equation*}
and
\begin{equation*}\delta'=v_g(n-t)+v_g(k).
\end{equation*}
\end{itemize}
\end{defi}
\vspace{0.3cm}

If there is no $t$ such that $t\prec'k$ and $v_{f}(t)=v_{f}(n-t)$,
then $f$ is contained in Class 1. If such $t$ exists, $f$ is
contained in Class 2. Class 2 is defined based on Class 1.

 \vspace{0.3cm}
\begin{theo} Suppose $f \in \text{SB}_n$, $n = 2k$. If $\text{AI}(f) = k$, then $f$ is in Class 1 or 2.
\end{theo}
\begin{proof}
If for any $t\prec'k$, $v_{f}(t)\not =v_{f}(n-t)$, we will prove $f$
is in Class 1. By Theorem \ref{thm_antij} and Theorem
\ref{thm_converse}, we know that $v_f(i) = v_f(j) + 1$ for all $i, j
\in A_p^k$ and $i < k < j$, which
 satisfies the definition of Class 1 functions.

 If there is some $t\prec'k$, $v_{f}(t) =v_{f}(n-t)$.
 By Theorem \ref{thm_atmostone}, we know at most one such $t$
 can exist.

 When $v_{f}(t) =v_{f}(n-t)=0$, by Theorem \ref{thm_vfk},
 we know $v_f(k) = 1$. Let $v_{g_1} = v_f + e_t$ and
 $v_{g_2} = v_f + e_{n-t}$. By Theorem \ref{thm_antij}
 and Theorem \ref{thm_converse}, $g_1, g_2$ are in Class 1 and
 $v_f = v_g + e_t + \delta e_k$ and
 $v_f = v_g + e_{n-t} + \delta' e_k$,
 where $\delta, \delta'$ are in Definition \ref{def_class}.
 When $v_{f}(t) =v_{f}(n-t)=1$, the proof is the same.
\end{proof}
\vspace{0.3cm}

Classes 1 and 2 consist of all functions satisfying the necessary
conditions to reach maximum algebraic immunity. In the following
sequel, we will prove that they do reach maximum AI, i.e., the
necessary conditions are sufficient.

\section{Functions in Class 1 Have Maximum Algebraic Immunity}
\vspace{0.3cm} Given a positive integer $k=(k_m, \ldots, k_1,
k_0)_2$ and any nonnegative integer $i$, we denote the vector
$$ (\varepsilon_{i,0}, \ldots, \varepsilon_{i,k-2}, \varepsilon_{i,k-1})\in \text{F}_2^k$$
by $\varepsilon_{i}^{k}$ or simply $\varepsilon_i$ if there is no
confusion, where
\begin{equation*}\varepsilon_{i,j}=
\begin{cases}
1,\hspace{0.1cm}\text{if}\hspace{0.1cm}j\preceq i,\\
0,\hspace{0.1cm}\text{otherwise}.\\
\end{cases}
\end{equation*}
Equivalently,
$$
 \varepsilon_i = \sum_{\substack{j \preceq i\\0
\leq j \leq k - 1}}{e_j}.
$$
Furthermore, the inverse representation is easy to obtain, namely,
\begin{equation}
\label{equm} e_i=\sum_{\substack{j \preceq i\\0 \leq j \leq
k-1}}{\varepsilon_j}.
\end{equation}
Therefore, $\{\varepsilon_0,$ $\varepsilon_1,$ $\ldots,$
$\varepsilon_{k-1}\}$ is a basis of $\text{F}_2^k$. Moreover,
$\{\varepsilon_{k+1},$ $\varepsilon_{k+2},$ $\ldots,$
$\varepsilon_{2k}\}$ is also a basis of $\text{F}_2^k$, as the
following lemma states.\vspace{0.3cm}
\begin{lem}
\label{lem_bases} $\{\varepsilon_0,$ $\varepsilon_1,$ $\ldots,$
$\varepsilon_{k-1}\}$ and $\{\varepsilon_{k+1},$
$\varepsilon_{k+2},$ $\ldots,$ $\varepsilon_{2k}\}$ are two bases of
$\text{F}_2^k$.
\begin{proof}
By (\ref{equm}), we claim
$\{\varepsilon_{0},\varepsilon_{1},...,\varepsilon_{k-1}\}$ is a
basis of $\text{F}^{k}_{2}$. For
$\{\varepsilon_{k+1},\varepsilon_{k+2},...,\varepsilon_{2k}\}$,
let's consider a system of homogeneous equation on variables
$x_{0},x_{1},...,x_{k-1}$:
$$\left( \begin{array} {c}\varepsilon_{k+1} \\
 \varepsilon_{k+2}  \\
 \vdots  \\
 \varepsilon_{2k}   \\
 \end{array}\right)X^\text{T}=0,$$
where $X=(x_{0},x_{1},...,x_{k-1})\in \text{F}^{k}_{2}$. We assume
that this equation system has a nonzero solution
$\lambda$=($\lambda_0,$ $\lambda_1,$ $\ldots,$ $\lambda_{k-1})\in
\text{F}_2^k$. Let $g(x)=\sum^{k-1}_{i=0}\lambda_{i}\sigma_{i}\in
\text{SB}_{n}$, then
$\lambda_{g}=(\lambda_{0},\lambda_{1},...,\lambda_{k-1},0,0,...,0)\in
\text{F}^{2k+1}_{2}.$ According to the assumption, for $k+1\leq
i\leq 2k$,
\begin{eqnarray*}
v_{g}(i)=\sum_{j\preceq i}\lambda_{g}(j)=\sum_{\scriptstyle j\preceq
i\atop\scriptstyle 0\leq j\leq k}\lambda_{j} =\varepsilon
_{i}\lambda^{T}=0,
\end{eqnarray*}
which means $v_{g}(i)=0$ holds for $k+1\leq i\leq 2k$. Let $f\in
\text{SB}_{n}$ be the function in Lemma 2.3. Thus, $fg=0$ and
$\deg(g)<k$.
\par However, by Lemma 2.3, $\text{AI}(f)=k$. Therefore, we have $g=0$,
so the above system can have only one solution $X=0$. Thus,
$\{\varepsilon_{k+1},$ $\varepsilon_{k+2},$ $\ldots,$
$\varepsilon_{2k}\}$ is a basis of $\text{F}_2^k$.
\end{proof}
\end{lem}\vspace{0.3cm}

For any $0\leq i\leq \lfloor\log_2k\rfloor$, let
$$
\text{U}_i = \{\varepsilon_j\mid j\in\text{A}_{i+1}^k, 0 \leq j \leq
k-1\},
$$
$$ \text{V}_i = \{\varepsilon_j\mid j\in\text{A}_{i+1}^k, k+1\leq j\leq 2k
\},
$$
and
$$
\text{W}_i\in \{\text{U}_{i}, \text{V}_{i}\}.
$$

\vspace{0.3cm}
\begin{lem} $\text{U}_0$ or $\text{V}_0$, union $\text{U}_1$ or $\text{V}_1$, $\ldots$, union
$\text{U}_{\lfloor\log_2{k}\rfloor}$ or
$\text{V}_{\lfloor\log_2{k}\rfloor}$, denoted by
$\bigcup_{i=0}^{\lfloor\log_2{k}\rfloor}\text{W}_{i}$, is a basis of
$\text{F}_2^k$.
\begin{proof}
First, we prove that all vectors in $\text{V}_p$ can be written as
linear combinations of vectors in $\bigcup_{i=0}^{p}{\text{U}_i}$,
i.e., $\text{V}_p \subseteq \spa(\bigcup_{i=0}^{p}{\text{U}_i})$.
Take an arbitrary vector in $\text{V}_p$, denoted by
$\varepsilon_t$, for some $t\in \text{A}_{p+1}$ and $k+1\leq t\leq
2k$. Then, $t=(*,$ $\overline{k_p},$ $\ldots,$ $k_1,$ $k_0)_2$ by
Lemma 3.2, where $*$ is a binary string of arbitrary length. We can
expand $\varepsilon_t$ as follows:
\begin{eqnarray*}
\varepsilon_t & = & \sum_{\substack{i\preceq t\\0 \leq i\leq
k-1}}{e_i}
 =\sum_{\substack{i\preceq t\\0 \leq i\leq k - 1}}{\sum_{j
\preceq i}{\varepsilon_j}} \\
& = & \sum_{\substack{j\preceq t\\0 \leq j\leq k - 1}}{
\varepsilon_j\sum_{\substack{j\preceq i\preceq t\\0 \leq i\leq
k-1}}{1}}.
\end{eqnarray*}
Therefore, $\varepsilon_t$ can be written as a linear combination of
vectors in $\bigcup_{i=0}^{\lfloor\log_2{k}\rfloor}\text{U}_i$. For
any $\varepsilon_j\not\in\bigcup_{i=0}^{p}{U_i}$, i.e., $j=(*, k_p,
\ldots, k_1, k_0)_2$ and $j\leq k-1$, we calculate the coefficient
of $\varepsilon_j$, which is
\begin{equation}
\label{equ_sum1} \sum_{\substack{j \preceq i \preceq t\\0 \leq i
\leq k-1}}{1}=\sum_{\substack{(*, k_p, \ldots, k_1, k_0)_2\preceq i
\preceq (*, \overline{k_p}, \ldots, k_1, k_0)_2\\0 \leq i\leq
k-1}}{1}.
\end{equation}
When $k_p = 1$, there is no $i$ that satisfies the constraints;
thus, equation \eqref{equ_sum1} is $0$. When $k_p = 0$, if there is
an $i = (*, 0, i_{p-1}, \ldots, i_2, i_1, i_0)_2$ that satisfies
constraints $(*, k_p, \ldots, k_1, k_0)_2 \preceq i \preceq (*,
\overline{k_p}, \ldots, k_1, k_0)_2$ and $i\leq k-1$, it's not hard
to see $i+2^p=(*,1, i_{p-1}, \ldots, i_2, i_1, i_0)_2$ also
satisfies the above constraints and vice versa. Therefore, all $1$s
counted in equation \eqref{equ_sum1} are in pairs; thus, equation
\eqref{equ_sum1} is $0$. Since all $\varepsilon_j\not\in\bigcup_{i =
0}^{p}{\text{U}_i}$ will not exist in the expansion of
$\varepsilon_t\in\text{V}_p$, we conclude that $\text{V}_p\subseteq
\spa(\bigcup_{i = 0}^{p}{\text{U}_i})$.

Second, we use math induction to prove that the vector space spanned
by $\bigcup_{i=0}^{p}{\text{U}_i}$ is that spanned by
$\bigcup_{i=0}^{p}\text{W}_{i}$, for $p=0$, 1, $\ldots$, $\lfloor
\log_2{k}\rfloor$. The induction parameter is $p$.

\textbf{Basis:} We claim that $\spa(\text{U}_0) = \spa(\text{V}_0)$.
\par By Lemma 4.1, there is no linear dependence in
$\text{U}_{0}$ and $\text{V}_{0}$, so
$\dim\spa(\text{U}_{0})=|\text{U}_{0}|=|\text{V}_{0}|=\dim\spa(\text{V}_{0})$.
Having considered that $\text{V}_{0}\subseteq \spa(\text{U}_{0})$
and both $\text{U}_{0}$ and $\text{V}_{0}$ are finite, we claim
$\spa(\text{U}_{0})=\spa(\text{V}_{0})$.


\textbf{Induction:} Assume it is true for $p=0, 1, \ldots, q-1$.
Claim it is also true for $p=q$.

Since $\text{V}_q \subseteq\spa(\bigcup_{i=0}^{q}{\text{U}_i})$, we
have
\begin{eqnarray*}
\spa(\bigcup_{i=0}^{q}{\text{U}_i}) & = & \spa(\bigcup_{i =
0}^{q}{\text{U}_i}\cup
\text{V}_q) \\
& = & \spa(\bigcup_{i=0}^{q-1}{\text{U}_i}\cup\text{U}_q \cup\text{V}_q) \\
& = & \spa(\bigcup_{i=0}^{q-1}{\text{V}_i}\cup\text{U}_q \cup\text{V}_q) \\
& \supseteq & \spa(\bigcup_{i = 0}^{q}{\text{V}_i})
\end{eqnarray*}
Notice that
$\bigcup^{q}_{i=0}\text{U}_{i}\subseteq\{\varepsilon_{0},...,\varepsilon_{k-1}\}$
and
$\bigcup^{q}_{i=0}\text{V}_{i}\subseteq\{\varepsilon_{k+1},...,\varepsilon_{n-1}\}$
due to the definition of $\text{U}_{i}$ and $\text{V}_{i}$. By Lemma
4.1, there is no linear dependence in
$\bigcup^{q}_{i=0}\text{U}_{i}$ and $\bigcup^{q}_{i=0}\text{V}_{i}$
, which means $\dim \spa (\bigcup_{i =
0}^{q}{\text{U}_i})=\sum^{q}_{i=0}|\text{U}_{i}| $
$=\sum^{q}_{i=0}|\text{V}_{i}|=$ $\dim$ $ \spa$ $( \bigcup_{i =
0}^{q}{\text{V}_i})$. Thus, we have $\spa(\bigcup_{i =
0}^{q}{\text{U}_i})$ $= \spa(\bigcup_{i = 0}^{q}{\text{V}_i})$
$=\spa(\bigcup_{i = 0}^{q-1}\text{W}_{i}\cup \text{U}_q)$
$=\spa(\bigcup_{i=0}^{q-1}\text{W}_{i}\cup \text{V}_q)$, which
completes the induction.


Therefore, $\spa( \bigcup_{i = 0}^{\lfloor \log_2{k}
\rfloor}\text{W}_{i})=\spa(\bigcup_{i=0}^{\lfloor \log_2{k}
\rfloor}{U_i})$. By Lemma \ref{lem_bases}, we know $\bigcup_{i=
0}^{\lfloor\log_2{k}\rfloor}{U_i}$ is a basis of $\text{F}_2^k$.
Therefore, $\bigcup_{i=0}^{\lfloor\log_2{k}\rfloor}\text{W}_{i}$ is
also a basis of $\text{F}_2^k$.
\end{proof}
\end{lem}\vspace{0.3cm}

In Theorem 4.1, we consider the symmetric annihilators of Boolean
functions in Class 1. We show that all Boolean functions in Class 1
have no symmetric annihilator with degree less than $k$. In Theorem
4.2, we show that all Boolean functions in Class 1 have maximum
$\textrm{AI}$.\vspace{0.3cm}

\begin{theo}
\label{thm_c1_noan} Let $n=2k$ and $f\in \text{SB}_n$. If
$v_f(i)=v_f(j)+1$, for any $i, j\in\text{A}_t^k$ with $0\leq
i<k<j\leq 2k$ and any $1\leq t\leq \lfloor\log_2n\rfloor$, then
there does not exist any nonzero $n$-variable symmetric Boolean
function $g$ with degree less than $k$, such that $fg=0$ or
$(f+1)g=0$.
\begin{proof}  Let $g(x)=\sum_{0\leq
i<k}\lambda_{i}\sigma_{i}\in \text{SB}_n$ and
$\lambda$=($\lambda_0,$ $\lambda_1,$ $\ldots,$ $\lambda_{k-1})\in
\text{F}_2^k$. If $fg=0$, then
$v_g(i)=\varepsilon_i\lambda^\text{T}=0$ holds when $v_f(i)=1$.
According to the condition of this theorem that $v_f(i)=v_f(j)+1$,
for any $i, j \in \text{A}_t$ with $0 \leq i<k<j \leq 2k$ and any
$1\leq t\leq \lfloor\log_2n\rfloor$, we can obtain a system of
homogeneous linear equations on variables $\lambda_{0}, \ldots,
\lambda_{k-1}$ of the form
$$\left( \begin{array} {c}\alpha_{0} \\
 \alpha_{1}  \\
 \vdots  \\
 \alpha_{k-1}  \\
 \beta   \\
 \end{array}\right)\lambda^\text{T}=0,$$
where $\alpha_{0},$ $\alpha_{1},$ \ldots, $\alpha_{k-1},$ $\beta\in
\text{F}_{2}^{n}$, $\{\alpha_{0},$ $\alpha_{1}$, $\ldots,$
$\alpha_{k-1}\}$=$\bigcup_{i =
0}^{\lfloor\log_2{k}\rfloor}\text{W}_{i}$ and
$$
\beta=\begin{cases} 0, \text{if} ~v_f(k)=0,\\
\varepsilon_k, \text{ otherwise}. \end{cases}$$ However, by Lemma
4.2, matrices of this kind have full rank. So we have $\lambda=0$;
thus, $g=0$.

Denote the system of homogeneous linear equations obtained by the
condition $fg=0$ by
\begin{equation}
 \label{eq1}\text{M}_f\lambda^\text{T}=0,
\end{equation}
i.e., let $\text{M}_f$ be the coefficient matrix of the system.
Formally, coefficient matrix $M_f$ is defined as follows:
\begin{equation}
\label{def_Mf}
M_{f}=\left( \begin{array} {c}\varepsilon_{i_{1}} \\
 \varepsilon_{i_{2}}  \\
 \vdots  \\
 \varepsilon_{i_{m}}  \\
 \end{array}\right),
 \end{equation}
where $\varepsilon_{j}$ is a row vector of $M_{f}$ if and only if
$v_{f}(j)=1$. The row vectors of $M_{f}$ are ordered by $0 \le
i_{1}<i_{2}<\cdots<i_{m} \le 2k$.

Similarly, if $(f+1)g=0$, then the coefficient matrix
$\text{M}_{f+1}$ of the system of homogeneous linear equations
$$\text{M}_{f+1}\lambda^\text{T}=0$$ also has full rank. Therefore, $g=0.$
\end{proof}
\end{theo}



 \vspace{0.3cm} In the following sequel, we only consider the
rank of $M_{f}$, which means the order of the row vectors of $M_{f}$
is not important. From the definition of $M_f$, we know
$\varepsilon_{\omega}$ is a row vector of $M_{f}$ if and only if
$v_{f}(\omega)=1$.

\vspace{0.3cm} \begin{theo} \label{thm_C1_suff}
 Let $n=2k$ and $f \in \text{SB}_n$. If
$v_f(i)=v_f(j)+1$, for any $i, j \in \text{A}_t^k$ with $0 \leq
i<k<j \leq 2k$ and any $1\leq t\leq \lfloor\log_2n\rfloor$, then
$\text{AI}(f)=k$.
\begin{proof}
Assume to the contrary that $\text{AI}(f)<k$. Then, there exists a
Boolean function $0\neq g\in \text{B}_n$ with degree less than $k$,
such that $fg=0$ or $(f+1)g=0$.
\par For the case $fg=0$, by
Lemma 2.4, there exists a symmetric Boolean function $0\neq h\in
\text{SB}_{n-2b}$ with $\deg(h)\leq \deg(g)-b<k-b$ for some integer
$0\leq b\leq k$, such that $fhP_b=0$. Let $f_{1}\in SB_{n-2b}$, be
defined as
$$
f_{1}(x_{2b+1}, \ldots, x_n)=f(0, 1, \ldots, 0, 1, x_{2b+1}, \ldots,
x_n).
$$
Then
\begin{equation}
\label{i=i+b} v_{f_{1}}(i)=v_f(i+b)
\end{equation}
 for any $0\leq i\leq n-2b$.
\begin{itemize}
\item[i)] On one hand, we claim that $f_{1}h=0$. If $f_{1}h \not= 0$, then there
exists an $i$ such that $i \in \text{WS}(f_{1}) \cap \text{WS}(h)$,
where $i\in\text{WS}(f_{1})$ implies $i + b \in \text{WS}(f)$ by
\eqref{i=i+b} and $i \in \text{WS}(h)$ implies $i+b \in
\text{WS}(hP_b)$ by the definition of $P_{b}$. Thus, $i + b \in
\text{WS}(f) \cap \text{WS}(hP_b)$, which is contradicted with
$fhP_b = 0$.
\item[ii)] On the other hand, we will show a contradiction by proving $f_{1}$
and $f_{1}+1$ do not have symmetric annihilators with degree less
than $k-b$. For any $i, j \in A_t^{k-b}$, $1\leq
t\leq\lfloor\log_2(n-2b)\rfloor$, $i < k-b < j$, we have $i+b, j+b
\in A_t^k$ and $i+b < k < j+b$ by Definition \ref{def_Atk}. By the
conditions in this theorem, we have $v_f(i+b) = v_f(j+b) + 1$, which
implies $v_{f_{1}}(i) = v_{f_{1}}(j) + 1$ by \eqref{i=i+b}. Then
$f_{1}$ is contained in Class 1. According to Theorem 4.1, $f_{1}$
or $f_{1}+1$ do not have symmetric annihilators with degree less
than $k$, which is contradicted with the existence of $h$.
\end{itemize}
Therefore, $f$ does not have nonzero annihilators with degree less
than $k$.

For the case $(f+1)g=0$, we can consider $f_{1}+1$ instead. By the
same argument above, we can prove that if $f+1$ does not have
nonzero annihilator with degree less than $k$. Therefore, we have
$\text{AI}(f)=k$.

\end{proof}
\end{theo}\vspace{0.3cm}

\section{Functions in Class 2 Have Maximum Algebraic Immunity}
\vspace{0.3cm} In this section, we will use the same notations as
the last section, such as $\varepsilon_k$, $\text{M}_f$, and so on.
We always assume $n=2k$ and $k=(k_m,$ $\ldots,$ $k_1,$ $k_0)_2$,
where $m=\lfloor\log_2k\rfloor$. We denote by $\text{supp}(k)=\{p |
k_{p}=1\}$, then $m\in \text{supp}(k)$.

We first present Lemmas 5.1, 5.2, 5.3, and 5.4. With these lemmas,
we study the annihilator of Boolean functions in Class 2. In Theorem
5.1, the symmetric annihilators of Boolean functions in Class 2 are
studied. In Theorem 5.2, all the annihilators of Boolean functions
are studied. \par The following Lemma plays an important role in the
proof of Lemma 5.2.\vspace{0.3cm}
\begin{lem} Given three constants $a, b, c\in \text{F}_2$, consider the following
system of inequalities on variable $t\in \text{F}_2$
\begin{equation}\label{eq0}
\begin{cases} a\leq b+t\\
t\leq c,
\end{cases}
\end{equation}
we have

(1) If $a=c$, or $(a, b, c)$=$(1,$ $1,$ $0)$, then equations
(\ref{eq0}) have one solution.

(2) If $a=0$ and $c=1$, then equations (\ref{eq0}) have two
solutions.

(3) If $(a, b, c)$=$(1,$ $0,$ $0)$, then equations (\ref{eq0}) do
not have any solution.
\begin{proof} It is easy to obtain the conclusions.
\end{proof}
\end{lem}\vspace{0.3cm}

\begin{lem} For any $0\leq p\leq m+1$, we have
\begin{equation}\label{eq2}\displaystyle{\sum_{\substack{j\in
\text{A}_p, j\preceq k
}}\varepsilon_j}=\displaystyle{\sum_{\substack{
 j\in \text{A}_p, j\preceq k }}\varepsilon_{n-j}}.
 \end{equation}
\begin{proof} When $p=0$, equation (\ref{eq2}) holds because both sides are
$\varepsilon_k$. Moreover, if $k_{p-1}=0$, i.e., $p-1\not\in$
$\text{supp}(k)$, then $\{j\mid$ $j\in \text{A}_p,$ $j\preceq
k\}$=$\emptyset$, equation (\ref{eq2}) also holds. Therefore, in
what follows, we always assume $p\neq 0$ and $k_{p-1}=1$.

For the left-hand side of equation (\ref{eq2}), we have
\begin{eqnarray*}
\sum_{\substack{j\in \text{A}_p, j\preceq k }}\varepsilon_j & = &
\sum_{\substack{j\in \text{A}_p, j\prec k}}\sum_{\substack{i\preceq
j}}e_i\\ & = & \sum_{\substack{i\prec k}}e_i\sum_{\substack{j\in
\text{A}_p, i\preceq j\prec k}}1.
\end{eqnarray*}
Let $i=(i_{m},$ $\ldots,$ $i_{1},$ $i_{0})_2$, $j=(j_{m},$ $\ldots,$
$j_{1},$ $j_{0})_2$. By Lemma 3.2, if $j\in \text{A}_p$ and
$j\preceq k$, then we have
$$(j_{p-1}, j_{p-2}, \ldots, j_1, j_0)=(\overline{k_{p-1}}, k_{p-2},
\ldots, k_{1}, k_{0}).$$ Therefore,
$$\sum_{\substack{i\preceq j\prec k \\ j\in \text{A}_p}}1=2^{\text{wt}(k_m,
\ldots, k_p)-\text{wt}(i_m, \ldots, i_p)}.$$ Thus
$$\sum_{\substack{i\preceq j\prec k \\ j\in \text{A}_p}}1=1$$ if and only
if $$i=(k_m, \ldots, k_p, \overline{k_{p-1}}, i_{p-2},\ldots, i_{1},
i_{0})_2$$ with $(i_{p-2}, \ldots, i_{1}, i_{0})\preceq (k_{p-2},
\ldots, k_{1}, k_{0})$. Hence, we have
$$\sum_{\substack{j\in \text{A}_p, j\preceq k
}}\varepsilon_j= \sum_{\substack{i=(k_m, \ldots, k_p,
\overline{k_{p-1}}, i_{p-2},\ldots, i_{1}, i_{0})_2\\(i_{p-2},
\ldots, i_{1}, i_{0})\preceq (k_{p-2}, \ldots, k_{1}, k_{0})}}e_i.$$

For the right-hand side of equation (\ref{eq2}), we have
\begin{eqnarray*}
\sum_{\substack{j\prec k, \\j\in \text{A}_p}}\varepsilon_{2k-j} & =
&\sum_{\substack{j\prec k, \\j\in
\text{A}_p}}\sum_{\substack{i\preceq 2k-j
\\ 0\leq i\leq k-1}}e_i\\
& = &\sum_{\substack{ 0\leq i\leq k-1}}e_i\sum_{\substack{i\preceq
2k-j\\j\prec k, j\in \text{A}_p}}1.
\end{eqnarray*}
If $j\in\text{A}_p$, then the last $p$ bits of $j$ and $2k-j$ are
both $(\overline{k_{p-1}},$ $k_{p-2},$ $\ldots,$ $k_{1},$ $k_{0})$.
Hence,
$$k-j=(k_m+j_m, \ldots, k_p+j_p, 1, 0, \ldots, 0)_2$$ and we can write $2k-j=$
$(k_{m}+s_{m},$ $k_{m-1}+j_{m}+s_{m-1},$ $\ldots,$
$k_{p+1}+j_{p+2}+s_{p+1},$ $k_{p}+j_{p+1}$, $\overline{j_p},$ $0,$
$k_{p-2},$ $\ldots,$ $k_{0})_2,$ where for any $q>p$, $s_q=1$ if and
only if $(k_{q-1},$ $j_q,$ $s_{q-1})$=$(1,$ $1,$ $1)$ or,
$k_{q-1}=0$ and $(j_q,$ $s_{q-1})$ $\neq$ $(0, 0)$. Note that the
additions are in $\text{F}_2$.

If
$$(i_{p-1}, i_{p-2},\ldots, i_{1}, i_{0})\not\preceq
(\overline{k_{p-1}}, k_{p-2}, \ldots, k_{1}, k_{0}),$$ then $\{j\mid
i\preceq 2k-j, j\prec k, j\in \text{A}_p\}=\emptyset$; thus,

$$\sum_{\substack{i\preceq
2k-j\\j\prec k, j\in \text{A}_p}}1=0.$$

Now assume $$(i_{p-1}, i_{p-2},\ldots, i_{1}, i_{0})\preceq
(\overline{k_{p-1}}, k_{p-2}, \ldots, k_{1}, k_{0}),$$ then
$i\preceq 2k-j,$ $j\prec k,$ $j\in \text{A}_p$ if and only if
$$(j_{p-1}, j_{p-2},
\ldots, j_{0})=(\overline{k_{p-1}}, k_{p-2}, \ldots, k_{0})$$ and
\begin{equation}\label{eq5}\begin{cases} (i_m, \ldots, i_p)\preceq (k_{m-1}+j_{m}+s_{m-1},
\ldots, k_p+j_{p+1}, \overline{j_p}), \\(j_{m}, \ldots,
j_{x})\preceq (k_{m}, \ldots, k_p).
\end{cases}
\end{equation} Equations (\ref{eq5}) are equivalent to the intersections of a series of
systems of inequalities as follows:
$$
\begin{cases} i_p\leq 1+j_p\\
j_p\leq k_p
\end{cases},
\begin{cases}
i_{p+1}\leq k_p+j_{p+1}\\
j_{p+1}\leq k_{p+1}
\end{cases},
\ldots,$$
$$\begin{cases}
i_m\leq k_{m-1}+j_{m}+s_{m-1}\\
j_{m}\leq k_m
\end{cases}.
$$
If $(i_m, \ldots, i_p)=(k_m, \ldots, k_p)$, by Lemma 5.1, each
system of inequalities has one and only one solution with respect to
$j_q$. Thus, we can conclude that there is one and only one solution
of $(j_{m}, \ldots, j_{0})_2$ satisfying $i\preceq 2k-j,$ $j\prec
k,$ $j\in \text{A}_p$. Therefore, the coefficient of $e_i$, where
$i=(k_m,$ $\ldots,$ $k_p,$ $\overline{k_{p-1}},$ $i_{p-2},$
$\ldots,$ $i_{1},$ $i_{0})_2$, equals 1.

Else if $(i_m,$ $\ldots,$ $i_p)$ $\neq$ $(k_m,$ $\ldots,$ $k_p)$,
then the set $\{q\mid k_q=1,$ $i_q=0\}$ is not empty. For any such
$q$, the corresponding system of inequalities with respect to $j_q$
has two solutions, according to Lemma 5.1. While for any other
systems, either one solution or no solution exists. Therefore, the
number of $j$'s satisfying $i\preceq 2k-j,$ $j\prec k,$ and $j\in
\text{A}_p$ is even. Therefore, the coefficient of $e_i$ equals 0.
As a result, we have
$$
\sum_{\substack{j\prec k, j\in \text{A}_p}}\varepsilon_{2k-j} =
\sum_{\substack{i=(k_m, \ldots, k_p, \overline{k_{p-1}},
i_{p-2},\ldots, i_{1}, i_{0})_2\\(i_{p-2}, \ldots, i_{1},
i_{0})\preceq (k_{p-2}, \ldots, k_{1}, k_{0})}}e_i.$$ So equation
(\ref{eq2}) holds. This finishes the proof of this theorem.
\end{proof}
\end{lem}\vspace{0.3cm}

\begin{lem} For any positive integer $k$, we have
$$\varepsilon_{k}=\sum_{j\prec k}\varepsilon_{j}.$$
\begin{proof}
Let $i=(i_{m},...,i_{0})_{2}$, $j=(j_{m},...,j_{0})_{2}$ and
$k=(k_{m},...,k_{0})_{2}$. Consider the equation
$$\sum_{j\preceq k}\varepsilon_{j}=\sum_{j\preceq k}\sum_{\substack{i\preceq j \\0\leq i\leq k-1}}e_{i}=\sum_{0\leq i\leq k-1}e_{i}\sum_{i\preceq j\preceq k}1.$$
Notice here, $i\not=k$. By the definition of $\preceq$, for each
$i\preceq k$, the number of $j$ such that $i\preceq j\preceq k$ is
$$\sum_{i\preceq j\preceq k}1=2^{\text{wt}(k_{m},...,k_{0})-\text{wt}(i_{m},...,i_{0})}.$$
Then, we have $\sum_{i\preceq j\preceq k}1$ is always even because
$i\not=k$, which means $\sum_{j\preceq k}\varepsilon_{j}=0.$ Thus,
we have $\varepsilon_{k}=\sum_{j\prec k}\varepsilon_{j}$.
\end{proof}
\end{lem}\vspace{0.3cm}

\begin{lem}For any positive integer $k$, we have
$$\varepsilon_{k}=\sum_{j\prec k}\varepsilon_{n-j}.$$
\begin{proof}
According to Lemma 5.3, we have
$$\varepsilon_{k}=\sum_{j\prec k}\varepsilon_{j}=\sum_{0\leq p\leq m+1}\sum_{j\in A_{p}, j\prec k}\varepsilon_{j}$$
because $A_{p}$ for all $0\leq p\leq m+1$ is a partition of
$\{0,1,...,n\}$. By applying Lemma 5.2 to the right-hand side of the
equation above, we have
$$\varepsilon_{k}=\sum_{0\leq p\leq m+1}\sum_{j\in A_{p},j\prec k}\varepsilon_{n-j}=\sum_{j\prec k}\varepsilon_{n-j}.$$
\end{proof}
\end{lem}

With these lemmas above, we have \vspace{0.3cm}

\begin{theo}  Suppose $f\in \text{SB}_n$ such that there exists a function $f'$ in Class 1 and an integer $t\prec'
k$ such that
\begin{itemize}
\item[]\hspace{1.7cm} $v_f=v_{f'}+e_{t}+\delta e_{k}$, or
\item[]\hspace{1.7cm} $v_f=v_{f'}+e_{n-t}+\delta'e_{k}$,
\end{itemize}
where
\begin{equation*}\delta=v_{f'}(t)+v_{f'}(k),\\
\end{equation*}
and
\begin{equation*}\delta'=v_{f'}(n-t)+v_{f'}(k).\\
\end{equation*}
Then both $f$ and $f+1$ have no nonzero symmetric annihilators with
degree less than $k$.
\begin{proof} Any $n$-variable symmetric Boolean function $g$ with
degree less than $k$ can be written as $g=\sum_{0\leq
i<k}\lambda_{i}\sigma_{i}$. If $fg=0$ or $(f+1)g=0$, similar to the
proof of Theorem 4.1, we can set up a system of homogeneous linear
equations on variables $\lambda_{i}$'s. If we can show that the
coefficient matrix $M_{f}$ has full rank, then the equation system
has only zero solution which means $\lambda_{g}=0$; thus, $g=0$. By
the definition of coefficient matrix \eqref{def_Mf}, $M_{f}$ is only
slightly different from that of equation $M_{f'}$, because
$v_f(i)=v_{f'}(i)$ holds except for $i=t,$ $k,$ $n-t$. Notice that
we only care about the rank of $M_{f}$ and $M_{f'}$. Since $f'$ is
in Class 1, we have $v_{f'}(t)$=$v_{f'}(n-t)+1$ and $M_{f'}$,
$M_{f'+1}$ both have full rank by the proof of Theorem
\ref{thm_c1_noan}. Let $a\in \text{F}_2$, if

$$(v_{f'}(t), v_{f'}(k), v_{f'}(n-t))=(a, a, a+1),$$ then
$$(v_f(t), v_f(k), v_f(n-t))=(a+1, a, a+1) {\rm~or~} (a, a+1, a).$$

\vspace{0.3cm} \textbf{Case 1}: a=0.
$$(v_{f'}(t), v_{f'}(k), v_{f'}(n-t))=(0, 0, 1).$$
Let $p_{1}$ be the integer satisfying $t\in A_{p_{1}}$, then for any
$i>k$ and $i\in A_{p_{1}}$, we have $v_{f'}(i)=1$ for the reason
$f'$ belongs to Class 1. Consider the difference between the
coefficient matrix $\text{M}_f$ and $M_{f'}$. By the definition of
coefficient matrix \eqref{def_Mf}, since $v_{f'}(n-t)=1$ and
$v_{f'}(t)=v_{f'}(k)=0$, $\varepsilon_{n-t}$ is a row vector of
$M_{f'}$ but $\varepsilon_t$ and $\varepsilon_k$ are not.

When $(v_f(t), v_f(k), v_f(n-t))=(1, 0, 1)$, the row vectors of
$M_{f}$ and $M_{f'}$ are all the same except that an extra
$\varepsilon_t$ is a row vector of $M_{f}$ by the definition of
coefficient matrix \eqref{def_Mf}. Then, $M_{f}$ has full rank as
$M_{f'}$ has full rank.

When $(v_f(t), v_f(k), v_f(n-t))=(0, 1, 0)$, the only difference
with $M_{f'}$ is that $\varepsilon_{k}$ is a row vector of $M_{f}$
but $\varepsilon_{n-t}$ is not by the definition of coefficient
matrix \eqref{def_Mf}. If we can prove that $\varepsilon_{n-t}$ is a
linear combination of $\varepsilon_{k}$ and other row vectors in
$M_{f}$, then $M_{f}$ also has full rank. In the next paragraph, we
will complete this proof.

By Lemma 5.4 and the fact that $A_{p}$ for all $0\leq p\leq m+1$ is
a partition of $\{0,1,...,n\}$, we have
\begin{equation}
\label{varepsilon}\sum_{j\in A_{p_{1}},j\preceq
k}\varepsilon_{n-j}=\varepsilon_{k}+\sum_{\substack{p\not=p_{1}\\
0\leq p\leq m+1}}\sum_{\substack{j\in A_{p}\\j\prec
k}}\varepsilon_{n-j}.
\end{equation}
Since $t\in A_{p_{1}}$ and $t\prec' k$ which means $t\preceq k$,
then $\varepsilon_{n-t}$ appears on the left-hand side of
\eqref{varepsilon}. Remember that for any $i>k$ where $i\in
A_{p_{1}}$, we have $v_{f'}(i)=1$. Thus, we have $v_{f}(i)=1$ for
any such $i$ with $i\not= n-t$, which means all vectors in the
left-hand side of \eqref{varepsilon} are row vectors of $M_{f}$
except $\varepsilon_{n-t}$ by the definition of (12). While for the
right-hand side of \eqref{varepsilon}, notice that since $f'$ is
contained in Class 1. Then if one element $\omega$ in the set
$\{j|j\in A_{p},j\prec k\}$ (or $\{n-i|j\in A_{p},j\prec k\}$)
satisfies $v_{f'}(\omega)=1$, so do the other elements in the set.
By the definition of\eqref{def_Mf}, if one term in
$\sum_{\substack{j\in A_{p}, j\prec k}}\varepsilon_{j}$ ( or
$\sum_{\substack{j\in A_{p}, j\prec k}}\varepsilon_{n-j}$) is a row
vector of $M_{f'}$, so are the other terms. By the definition of
$M_{f}$, we can see $M_{f}$ also has this property as $M_{f'}$ for
any $p\not=p_{1}$. Then, we can turn all terms on the right-hand
side of \eqref{varepsilon} into the row vectors of $M_{f}$ by Lemma
5.2 as follows. If $\sum_{\substack{j\in A_{p}, j\prec
k}}\varepsilon_{n-j}$ with $p\not=p_{1}$ appears but not row vectors
of $M_{f}$, then we apply Lemma 5.2 to replace $\sum_{j\in A_{p},
j\prec k}\varepsilon_{n-j}$ by $\sum_{j\in A_{p}, j\prec
k}\varepsilon_{j}$. We can continue this procedure until that any
vector that appears in the right-hand side also appears as a row
vector of $\text{M}_f$.

After this, both sides of \eqref{varepsilon} become row vectors of
$M_{f'}$ except $\varepsilon_{n-t}$. Thus we can conclude that
$\varepsilon_{n-t}$ is a linear combination of row vectors of
$\text{M}_f$. Therefore, $\text{M}_f$ has full rank because $M_{f'}$
has full rank.



Consider the difference between $\text{M}_{f+1}$ and
$\text{M}_{f'+1}$. Since
$$(v_{f'+1}(t), v_{f'+1}(k), v_{f'+1}(n-t))=(1, 1, 0),$$
and recalling the definition of coefficient matrix \eqref{def_Mf},
$\varepsilon_t$ and $\varepsilon_k$ are row vectors of
$\text{M}_{f'+1}$. While for $\text{M}_{f+1}$, if $(v_{f+1}(t),
v_{f+1}(k), v_{f+1}(n-t))=(0, 1, 0)$, the only difference is that
$\varepsilon_t$ is not a row vector of $M_{f+1}$. By the same
argument above using Lemmas 5.2 and 5.3, $\varepsilon_t$ is a linear
combination of row vectors of $\text{M}_{f+1}$. This has no
influence to the rank of $\text{M}_{f+1}$. Therefore,
$\text{M}_{f+1}$ also has full rank because $M_{f'+1}$ has full
rank.

Otherwise, if $(v_{f+1}(t), v_{f+1}(k), v_{f+1}(n-t))=(1, 0, 1)$,
according to the definition of coefficient matrix \eqref{def_Mf},
the difference is that $\varepsilon_{n-t}$ is a row vector of
$M_{f+1}$ but $\varepsilon_{k}$ is not. This matrix is also full
rank because $\varepsilon_{k}$ is a linear combination of other row
vectors in $M_{f'}$ due to Lemma 4.2.

\vspace{0.3cm} \textbf{Case 2}: a=1.

We can reach the conclusion that the matrices $M_f$ and $M_{f+1}$
always have full rank in the same way of \textbf{Case 1}.

Else if $$(v_{f'}(t), v_{f'}(k), v_{f'}(n-t))=(a, a+1, a+1),$$ then
$$(v_f(t), v_f(k), v_f(n-t))=(a+1, a, a+1) {\rm~ or~} (a, a+1, a).$$

In a similar way to the discussion as above, no matter $a=0$ or 1
and no matter what $f'$ is, the coefficient matrices $\text{M}_f$
and $\text{M}_{f+1}$ always have full rank, which implies
$\lambda=0$; thus, $g=0$.

Therefore, both $f$ and $f+1$ have no symmetric annihilators with
degree less than $k$.
\end{proof}
\end{theo}\vspace{0.3cm}

\begin{theo}
\label{thm_C2_suff} Suppose $f\in \text{SB}_n$ such that there
exists a function $f'$ in Class 1 and an integer $t\prec' k$, such
that
\begin{itemize}
\item[]\hspace{1.7cm} $v_f=v_{f'}+e_{t}+\delta e_{k}$, or
\item[]\hspace{1.7cm} $v_f=v_{f'}+e_{n-t}+\delta' e_{k}$,
\end{itemize}
where
\begin{equation*}\delta=v_{f'}(t)+v_{f'}(k),\\
\end{equation*}
and
\begin{equation*}\delta'=v_{f'}(n-t)+v_{f'}(k).\\
\end{equation*}
Then $\text{AI}(f)=k$.

\begin{proof}
Assume to the contrary that $\text{AI}(f)<k$. Then, there exists a
Boolean function $0\neq$ $g\in$ $\text{B}_n$ with degree less than
$k$, such that $fg=0$ or $(f+1)g=0$.
\par For the case $fg=0$, by Lemma 2.4, there exists a
symmetric Boolean function $0\neq$ $h\in$ $\text{SB}_{n-2b}$ with
$\deg(h)$ $\leq$ $\deg(g)-b<k-b$ for some integer $0\leq b\leq k$,
such that $fhP_b=0$.

Let $f_{1}\in SB_{n-2b}$ be defined as
$$
f_{1}(x_{2b+1}, \ldots, x_n)=f(0, 1, \ldots, 0, 1, x_{2b+1}, \ldots,
x_n).
$$
Then
\begin{equation}
\label{class_all}v_{f_{1}}(i)=v_f(i+b)
\end{equation}
for any $0\leq i\leq n-2b$.

On one hand, by following the same argument of i) in the proof of
Theorem \ref{thm_C1_suff}, we have $f_{1}h=0$ because $fhP_b=0$.

On the other hand, we will show a contradiction by proving $f_{1}$
and $f_{1}+1$ do not have symmetric annihilators with degree less
than $k-b$. For any $i, j \in A_p^{k-b}$, $1\leq
p\leq\lfloor\log_2(n-2b)\rfloor$, $i < k-b < j$, we have $i+b, j+b
\in A_p^k$ and $i+b < k < j+b$ by Definition \ref{def_Atk}.

If $t<b$, then $t-b$, $n-(t-b)$ $\not\in A_{p}^{k-b}$ for any $1\leq
p\leq\lfloor\log_2(n-2b)\rfloor$. Then, by the conditions of this
theorem, we have $v_{f}(i+b)=v_{f}(j+b)+1$, which implies
$v_{f_{1}}(i)=v_{f_{1}}(j)+1$ by \eqref{class_all}. Then, $f_{1}$ is
contained in Class 1. According to Theorem 4.1, $f_{1}$ and
$f_{1}+1$ do not have symmetric annihilators with degree less than
$k$, which contradicts the existence of $h$.

If $t\geq b$, then $t-b\prec'k-b$ by the definition of $\prec'$ and
$t-b$, $n-(t-b)$ $\in A_{p}^{k-b}$ for some $1\leq
p\leq\lfloor\log_2(n-2b)\rfloor$. Then, by the conditions of this
theorem, we have $v_{f}(t)=v_{f}(n-t)=v_{f}(k)+1$ and
$v_{f}(i+b)=v_{f}(j+b)+1$ for $i+b\not=t$ and $j+b\not=n-t$, which
implies $v_{f_{1}}(t-b)=v_{f_{1}}(n-(t-b))=v_{f_{1}}(k)+1$ and
$v_{f_{1}}(i)=v_{f_{1}}(j)+1$ for $i\not=t-b$ and $j\not=n-(t-b)$.
Then, $f_{1}$ is contained in Class 2. According to Theorem 5.1,
$f_{1}$ and $f_{1}+1$ do not have symmetric annihilators with degree
less than $k$, which contradicts the existence of $h$.



Therefore, $f$ does not have nonzero annihilators with degree less
than $k$.

For the case $(f+1)g=0$, we can consider $f_{1}+1$ instead. By the
same argument above, we can prove that if $f+1$ does not have
nonzero annihilator with degree less than $k$. Therefore, we have
$\text{AI}(f)=k$.
\end{proof}
\end{theo}\vspace{0.3cm}

\section{Main Result}
\vspace{0.3cm} Finally, we obtain the following main
result.\vspace{0.3cm}
\begin{theo} Let $n=2k$, $k=(k_m,$ $\ldots,$ $k_1,$ $k_0)_2$,
where $m=\lfloor\log_2k\rfloor$. Let $\text{supp}(k)=\{p\mid
k_p=1\}$. Given $f\in \text{SB}_n$, then $\text{AI}(f)=k$ if and
only if $v_f$ satisfies one of the following three cases:

1) There exist $a_0,$ $a_1,$ $a_2,$ $\ldots,$ $a_m,$ $a_{m+1}\in
\text{F}_2$, such that for any $1\leq t\leq m+1$ and $i, j \in A_t,
0 \le i < k < j \le 2k, v_f(i) = v_f(j) + 1$ =$a_t$ holds, and
$v_f(k)=a_0$.

2) There exist $a_1,$ $a_2,$ $\ldots,$ $a_m,$ $a_{m+1}\in
\text{F}_2$, such that for any $1\leq t\leq m$ and $i, j \in A_t, 0
\le i < k < j \le 2k, v_f(i) = v_f(j) + 1$=$a_t$ holds, and
$$(v_f(k-2^m), v_f(k),
v_f(k+2^m)))=(a_{m+1}, a_{m+1}+1, a_{m+1});$$

3) There exists an integer $p_0\in\text{supp}(k)$, $p_0\neq m$, and
$a_1,$ $a_2,$ $\ldots,$ $a_m,$ $a_{m+1},$ $b_{p_0}\in\text{F}_2$
such that for any $1\leq t\leq m+1$, $i, j \in A_t, 0 \le i < k < j
\le 2k$, $i\not=i_{0}$ and $j\not=n-i_{0}$, $ v_f(i) = v_f(j) +
1$=$a_t$ holds, and
$$(v_f(i_0), v_f(k), v_f(n-i_0))=(b_{p_0},
\overline{b_{p_0}}, b_{p_0}), {\rm~or~}$$
$$(v_f(i_0), v_f(k), v_f(n-i_0))=(\overline{b_{p_0}}, b_{p_0}, \overline{b_{p_0}}),$$
where $i_0$=$(k_{p_0-1},$ $\ldots,$ $k_1,$ $k_0)\prec'k$.

The total number is $(2\wt(n)$+$1)2^{\lfloor\log_2 n \rfloor}$. And
the values of their Hamming weight are
$\{2^{n-1}$$\pm$$\frac{1}{2}\binom{n}{n/2},$
$2^{n-1}$+$\frac{1}{2}\binom{n}{n/2}$-$\binom{n}{i},$
$2^{n-1}$-$\frac{1}{2}\binom{n}{n/2}$+$\binom{n}{i},$ ${\rm ~for~
any~}i\prec'k\}$.
\begin{proof}
Notice that $v_f$ satisfies item 1 if and only if $f$ belongs to
Class 1, $v_f$ satisfies items 2 or 3 if and only if $f$ belongs to
Class 2. Then, the necessity is proved by Theorem 3.5, while the
sufficiency is proved by Theorems 4.2 and 5.2.

The number of functions satisfying item 1 is $2^{m+2}$, the number
of functions satisfying item 2 is $2^{m+1}$, and the number of
functions satisfying item 3 is $(\wt(n)-1)2^{m+2}$, because
different choices of $a_i$'s, $b_{p_0}$ and $i_0$ will generate
different functions. Therefore, the total number is
$(2\wt(n)+1)2^{m+1}$. The corresponding values of Hamming weight are
also easy to calculate.
\end{proof}
\end{theo}\vspace{0.3cm}

For example, when $n=14$, we have $k=7$, $m=\lfloor\log_27\rfloor=2$
and $\text{supp}(7)$=\{0, 1, 2\}.

There are 16 functions satisfying the conditions of item 1). Among
them, the simplified value vectors of those satisfying $v_f(7)=1$
are as follows:
\begin{table}[H]
\centering
\begin{tabular}[htbp]{c|c}
\hline
$f$& \text{SVV}:$v_{f}(0)$...$v_{f}(14)$\\
\hline
1&000000011111111\\
\hline
2&000100011110111\\
\hline
3&010001011011101\\
\hline
4&101010110101010\\
\hline
5&010101011010101\\
\hline
6&101110110100010\\
\hline
7&111011110001000\\
\hline
8&111111110000000\\
\hline
\end{tabular}
\end{table}
All of them have the same Hamming weight
$2^{13}$+$\frac{1}{2}\binom{14}{7}$=$9908,$ whereas all of their
complements have the same Hamming weight
$2^{13}$-$\frac{1}{2}\binom{14}{7}$=$6476.$

There are 8 functions satisfying the conditions of item 2. Among
them, the simplified value vectors of those satisfying $v_f(7)=1$
are as follows:
\begin{table}[H]
\centering
\begin{tabular}[htbp]{c|c}
\hline
$f$& \text{SVV}:$v_{f}(0)$...$v_{f}(14)$\\
\hline
1&000000011110111\\
\hline
2&010001011010101\\
\hline
3&101010110100010\\
\hline
4&111011110000000\\
\hline
\end{tabular}
\end{table}
All of them have the same Hamming weight
$2^{13}$+$\frac{1}{2}\binom{14}{7}$-$\binom{14}{3}$=$9544,$ whereas
all of their complements have the same Hamming weight
$2^{13}$-$\frac{1}{2}\binom{14}{7}$+$\binom{14}{3}$=$6840.$

There are 32 functions satisfying the conditions of item 3 because
$\wt(14)=3$ and $i_0$ can be 0 or 1. When $i_0=0$, among such
functions, the simplified value vectors of those satisfying
$v_f(7)=1$ are as follows:
\begin{table}[H]
\centering
\begin{tabular}[htbp]{c|c}
\hline
$f$& \text{SVV}:$v_{f}(0)$...$v_{f}(14)$\\
\hline
1&000000011111110\\
\hline
2&010001011011100\\
\hline
3&001010110101010\\
\hline
4&011011110001000\\
\hline
5&000100011110110\\
\hline
6&010101011010100\\
\hline
7&001110110100010\\
\hline
8&011111110000000\\
\hline
\end{tabular}
\end{table}
All of them have the same Hamming weight
$2^{13}$+$\frac{1}{2}\binom{14}{7}$-$\binom{14}{0}$=$9907,$ whereas
all of their complements have the same Hamming weight
$2^{13}$-$\frac{1}{2}\binom{14}{7}$+$\binom{14}{0}$=$6477.$

When $i_0=1$, among such functions, the simplified value vectors of
those satisfying $v_f(7)=1$ are as follows:
\begin{table}[H]
\centering
\begin{tabular}[htbp]{c|c}
\hline
$f$& \text{SVV}:$v_{f}(0)$...$v_{f}(14)$\\
\hline
1&000000011111101\\
\hline
2&000001011011101\\
\hline
3&101010110101000\\
\hline
4&101011110001000\\
\hline
5&000100011110101\\
\hline
6&000101011010101\\
\hline
7&101110110100000\\
\hline
8&101111110000000\\
\hline
\end{tabular}
\end{table}
All of them have the same Hamming weight
$2^{13}$+$\frac{1}{2}\binom{14}{7}$-$\binom{14}{1}$=$9894,$ whereas
all their complements have the same Hamming weight
$2^{13}$-$\frac{1}{2}\binom{14}{7}$+$\binom{14}{1}$=$6490.$
\section{Conclusion}
\vspace{0.3cm} In this paper, we give a necessary and sufficient
condition for an even-variable symmetric Boolean function to reach
maximum algebraic immunity for the first time.

We first study the weight supports of low-degree symmetric Boolean
functions and use some linear algebras to obtain some necessary
conditions for an even-variable symmetric Boolean function to reach
maximum algebraic immunity, then we divide the functions satisfying
these conditions into two classes. Finally, we proved that functions
of either class indeed have maximum algebraic immunity. Thus, the
problem of finding all even-variable symmetric Boolean functions
with maximum algebraic immunity is solved. \vspace{0.3cm}


%



\section*{Acknowledgment}
\par The authors are grateful to
the anonymous reviewers for their valuable comments and kind
suggestions that have improved the technical quality and editorial
presentation of this paper.

\ifCLASSOPTIONcaptionsoff
  \newpage
\fi



%

%

\begin{IEEEbiography}{Hui Wang}
was born in Anhui, China in 1988. He is currently working towards
the Ph.D. degree at Fudan University, China.
\end{IEEEbiography}






\end{document}